\documentclass{article}
\usepackage[utf8]{inputenc}
\usepackage{latexsym}
\usepackage{amsmath}
\usepackage{amssymb}
\usepackage{amsthm}
\usepackage{bbm}
\usepackage{graphicx}
\usepackage{pictexwd,dcpic}
\usepackage{subfigure}
\usepackage{dcolumn} 
\usepackage{braket}
\usepackage{bm}
\usepackage{cite}
\usepackage{yhmath}
\usepackage{import}
\usepackage{xifthen}
\usepackage{transparent}
\usepackage{geometry}
\usepackage{authblk} 
\usepackage{mathtools}

\usepackage[usenames,dvipsnames]{color}
\usepackage[]{hyperref}

\usepackage{cleveref}
\usepackage{youngtab}
\usepackage{tikz}
\usepackage{tikz-cd}
\usetikzlibrary{ decorations.markings}

\usepackage{geometry}

\usetikzlibrary{arrows}
\usepackage{fancyhdr}
\usepackage{xcolor}
\newsavebox\MBox

\usepackage{setspace}
\usepackage{a4wide}
\usepackage{amsthm,array,booktabs}
\usepackage{cancel}
\usepackage{tikz}
\usepackage{amscd}
\usepackage{latexsym,cite}
\usepackage{ytableau}
\usepackage{scalerel}
\usepackage{calc}
\input xy
\xyoption{all}
\hypersetup{
  unicode,
  bookmarksnumbered,
  linktoc = all,
  pdfborderstyle = {/S/U/W 0.5}
}

\newtheorem{remark}{Remark}

\newtheorem{thm}{Theorem}[section]
\newtheorem{cor}[thm]{Corollary}
\newtheorem{lem}[thm]{Lemma}

\newtheorem{example}{Example}

\numberwithin{equation}{section}

\title{\textbf{Quantum integrable model for the quantum cohomology/K-theory of flag varieties and the double $\beta$-Grothendieck polynomials}}
\author{Jirui Guo}
\affil{School of Mathematical Sciences,\protect\\
Institute for Advanced Study,\protect\\ 
Key Laboratory of Intelligent Computing and Applications,\protect\\
Tongji University, Shanghai 200092, China}
\date{jrkwok@tongji.edu.cn}

\begin{document}
\maketitle

\begin{abstract}
\noindent
A GL$(n)$ quantum integrable system generalizing the asymmetric five vertex spin chain is shown to encode the ring relations of the equivariant quantum cohomology and equivariant quantum K-theory ring of the flag varieties. We also show that the Bethe Ansatz states of this system generate the double $\beta$-Grothendieck polynomials.
\end{abstract}
\newpage
\tableofcontents
\setcounter{footnote}{0}

\section{Introduction}

The mathematical structure and exact solutions of quantum integrable systems have proved to be useful tools in providing information about physical phenomena, as well as revealing certain algebraic results. For example, in \cite{GK}, it was shown that the underlying Frobenius algebra of the asymmetric five vertex spin chain, which depends on a parameter $\beta$, is exactly the equivariant quantum cohomology of the Grassmannians when $\beta = 0$, and the equivariant quantum K-theory ring of the Grassmannians when $\beta = -1$. The Bethe/Gauge correspondence of this model was later studied in \cite{UY}. The present work aims at generalizing the idea of \cite{GK}, establishing a relationship between a GL$(n)$ integrable model and the equivariant quantum cohomology/quantum K-theory ring of general flag varieties.

The relationship between the quantum cohomology/K-theory of flag varieties and Toda lattices was discussed in \cite{GK93} and \cite{GL01}. Moreover, it was shown in \cite{KPSZ} and \cite{Huq-Kuruvilla:2024tsg} that the quantum K-theory ring of the cotangent bundle of the flag varieties can be realized by the XXZ model, and in the compact limit it reduces to the quantum K-theory ring of the flag varieties. In this paper, we show that the Bethe Ansatz equations of the GL$(n)$ five vertex model, which depends on a parameter $\beta$, give rise to the quantum Whitney sum relations of the equivariant quantum cohomology ($\beta=0$) and quantum K-theory ring ($\beta = -1$) of flag varieties. These quantum Whitney sum relations express the ring relations in terms of the Chern classes and the $\lambda_y$ classes of the tautological bundles in the cohomology and K-theory rings respectively. In the case of quantum cohomology, these relations are equivalent to those proved in \cite{AS95, K95}. While in the case of quantum K-theory, these relations were proposed in \cite{GMSXZZ} based on field-theoretic computations and were later proved for the Grassmannians \cite{GMSZ} and the incidence varieties \cite{Gu:2023fpw}. The Whitney sum relations for the quantum K-theory of general flag varieties were proved in \cite{Huq-Kuruvilla:2024tsg} using the abelian/non-abelian correspondence.

We will also show that the Bethe Ansatz states of the GL$(n)$ five vertex model generate the double $\beta$-Grothendieck polynomials defined in \cite{FK93}. When $\beta=0$ these polynomials reduce to the double Schubert polynomials representing Schubert classes in the equivariant cohomology ring, and when $\beta = -1$ they reduce to the double Grothendieck polynomials representing Schubert classes in the equivariant K-theory ring. More precisely, we show that, in the case of complete flag variety, the expansion coefficients of the Bethe Ansatz states in the natural basis are the double $\beta$-Grothendieck polynomials with variables being the Chern classes/K-theory classes of the line bundles $\mathcal{S}_i/\mathcal{S}_{i-1}$, where $\mathcal{S}_i$ is the $i$th tautological bundle. The double $\beta$-Grothendieck polynomials can also be realized as the partition functions of certain 2d lattice models, see e.g. \cite{BSK, BS, BFHTW}.

Our main results can be summarized as follows:
\begin{itemize}
\item The Bethe Ansatz equations of the GL$(n)$ quantum integrable system defined by the R-matrix $R^{(n)}(x,y)$ are derived (Theorem \ref{BAE}), where
\[
R^{(2)}(x,y) = \left( \begin{array}{cccc}
1 & 0 & 0 & 0 \\
0 & 0 & 1 & 0 \\
0 & 1+\beta(x \ominus y) & x\ominus y & 0 \\
0 & 0 & 0 & 1
\end{array}
\right)
\]
with $x \ominus y = \frac{x-y}{1+\beta y}$, and $R^{(n)}(x,y)$ is defined iteratively by
\[
R^{(n)}(x,y) = \left( \begin{array}{cccc}
1 & 0 & 0 & 0 \\
0 & 0 & I_{n-1} & 0 \\
0 & (1+\beta(x \ominus y))I_{n-1} & (x \ominus y) I_{n-1} & 0 \\
0 & 0 & 0 & R^{(n-1)}(x,y)
\end{array}
\right)
\]
for $n=3,4,\cdots$.
\item It is shown that the Bethe Ansatz equations give rise to the Whitney type ring relations of the equivariant quantum cohomology ring of the partial flag variety $\mathrm{Fl}(k_{n-1},\cdots,k_2,k_1;N)$ when $\beta=0$ in terms of the Chern classes (Eq. \eqref{whitneyQC}):
\[
c(\mathcal{S}_i) \cdot c(\mathcal{S}_{i+1}/\mathcal{S}_i) = c(\mathcal{S}_{i+1}) + (-1)^{k_{n-i-1}-k_{n-i}} q_{n-i} c(\mathcal{S}_{i-1}),
\]
and give rise to the ring relations of the equivariant quantum K-theory ring of $\mathrm{Fl}(k_{n-1},\cdots,k_2,k_1;N)$ when $\beta=-1$ in terms of the $\lambda_y$ classes (Eq. \eqref{whitneyQK}):
\[
\lambda_y(\mathcal{S}_i) * \lambda_y(\mathcal{S}_{i+1}/\mathcal{S}_i) = \lambda_y(\mathcal{S}_{i+1}) + \frac{q_{n-i}}{1-q_{n-i}} y^{k_{n-i-1}-k_{n-i}}\left( \lambda_y(\mathcal{S}_{i-1}) -\lambda_y(\mathcal{S}_i) \right) * \det(\mathcal{S}_{i+1}/\mathcal{S}_i),
\]
where $*$ stands for the product of the quantum K-ring, and $\mathcal{S}_i$ is the $i$th tautological bundle of the flag variety.
\item Given the matrix elements $B_i(u)$ of the monodromy matrix as in Eq. \eqref{Tab}, and the pseudo vacuum state $\ket{\Omega^{(0)}}$ as in \eqref{ground}, the expansion coefficients of the quantum state $B_{N-1}(\sigma_1) \cdots B_{1}(\sigma_{N-1}) \ket{\Omega^{(0)}}$ in the natural basis \eqref{basis} are the double $\beta$-Grothendieck polynomials $\mathcal{G}^{(\beta)}_\pi$ 
(Theorem \ref{thm1}): 
\[
B_{N-1}(\sigma_1)\cdots B_2(\sigma_{N-2}) B_1(\sigma_{N-1}) \ket{\Omega^{(0)}} = \sum_{\pi \in S_N} \mathcal{G}^{(\beta)}_\pi (\sigma_1,\cdots,\sigma_{N-1};\ominus t_1,\cdots,\ominus t_{N-1}) \ket{\omega_N \pi^{-1}},
\]
where $\omega_N$ is the permutation with maximal length in the permutation group $S_N$, and $\ket{w}$ for $w \in S_N$ is defined by Eq. \eqref{ketstate}.
\item When the elementary symmetric polynomial $e_l(\sigma^{(i)}_1,\cdots,\sigma^{(i)}_{N-i})$ is identified with $e_l(x_1,\cdots,x_{N-i})$ for all $i=1,\cdots,N-1$ and $l=1,\cdots,N-i$, the Bethe Ansatz state \eqref{BEstate} (in the complete flag case) has the expansion (Theorem \ref{thm2}):
\[
\begin{aligned}
\ket{\psi^{(0)}} &= \sum_{i_1,\cdots,i_{k_1}} \psi_{i_1\cdots i_{k_1}}^{(1)} B_{i_1}(\sigma_1^{(1)}) \cdots B_{i_{k_1}}(\sigma_{k_1}^{(1)}) \ket{\Omega^{(0)}}\\ &= \sum_{w \in S_N} \mathcal{G}^{(\beta)}_w (x_1,\cdots,x_{N-1};\ominus t_1,\cdots,\ominus t_{N-1}) \ket{\omega_N w^{-1}},
\end{aligned}
\]
where $x_i$ can be interpreted as the first Chern class of the line bundle $\mathcal{S}_i/\mathcal{S}_{i-1}$ when $\beta=0$, and $1-x_i$ can be interpreted as the K-theory class of $\mathcal{S}_i/\mathcal{S}_{i-1}$ when $\beta=-1$ (this means $\sigma^{(N-i)}_a$ ($\beta = 0$) or $1-\sigma^{(N-i)}_a$ ($\beta = -1$), $a=1,\cdots, i$, should be interpreted as the Chern roots of $\mathcal{S}_i$).
\end{itemize}
This paper is organized as follows. In Sec. \ref{sec:section2}, the basic ingredients of the GL$(n)$ quantum integrable system generalizing the asymmetric five vertex spin chain are introduced. In Sec. \ref{sec:QK}, the Bethe Ansatz equations are derived, which are shown to give rise to the Whitney type ring relations of the equivariant quantum cohomology/K-theory ring of the flag varieties. Sec. \ref{sec:Groth} consists of a series of propositions leading to our main results on the relationship between the quantum integrable model under study and the double $\beta$-Grothendieck polynomials. In the appendix, we present some background on the quantum cohomology/K-theory of flag varieties and the double $\beta$-Grothendieck polynomials.

\section{\label{sec:section2}The GL$(n)$ asymmetric five vertex model}

In this section, we present the basic ingredients of the GL$(n)$ asymmetric five vertex model. Let $V=\mathbb{C}^n$, we define the R-matrix $R^{(n)}(x,y) \in \mathrm{End}(V \otimes V)$ as follows:\\
For $n=2$, the R-matrix is defined as
\begin{equation}\label{R2}
R^{(2)}(x,y) = \left( \begin{array}{cccc}
1 & 0 & 0 & 0 \\
0 & 0 & 1 & 0 \\
0 & 1+\beta(x \ominus y) & x\ominus y & 0 \\
0 & 0 & 0 & 1
\end{array}
\right),
\end{equation}
where $\beta$ is a formal variable, and\footnote{This is for ease of notation. The operations $\oplus$ and $\ominus$ were defined in \cite{FK94}, where $x \oplus y = x+y+\beta x y$. It was noticed that for a generator $u_i$ of the IdCoxeter algebra, which satisfies $u_i^2 = \beta u_i$, one has $e^{t u_i} = 1+ [t]_\beta u_i$, where $[t]_\beta = (e^{\beta t} - 1)/\beta$,  and consequently $[t]_\beta \oplus [s]_\beta = [t+s]_\beta$ and $[t]_\beta \ominus [s]_\beta = [t-s]_\beta$.}
\begin{equation}\label{ominus}
x \ominus y \equiv \frac{x-y}{1+\beta y}.
\end{equation}
For $n>2$, under the decomposition $V \cong \mathbb{C} \oplus \mathbb{C}^{n-1}$, $R^{(n)}(x,y)$ is defined as follows:
\begin{equation}\label{Rmatrix}
R^{(n)}(x,y) = \left( \begin{array}{cccc}
1 & 0 & 0 & 0 \\
0 & 0 & I_{n-1} & 0 \\
0 & (1+\beta(x \ominus y))I_{n-1} & (x \ominus y) I_{n-1} & 0 \\
0 & 0 & 0 & R^{(n-1)}(x,y)
\end{array}
\right),
\end{equation}
where $I_r$ stands for the identity matrix of dimension $r$. The entries of the R-matrix can be written as
\begin{equation}\label{Rijkl}
R^{(n)}_{ij,kl}(x,y)= (1+ \delta_{k<l}\beta (x\ominus y))\delta_{il}\delta_{jk}+\delta_{k>l}(x\ominus y) \delta_{ik}\delta_{jl},
\end{equation}
where $i,j$ are row indices, $k,l$ are column indices, $\{i,k\}$ and $\{ j,l \}$ are indices for the first and second factor of $V \otimes V$ respectively. We also define
\[
\delta_{i<j} = \left\{ \begin{array}{ll}
1, & ~\mathrm{if}~ i<j,\\
0, &~\mathrm{if}~ i\geq j,
\end{array}
\right.\quad
\delta_{i>j} = \left\{ \begin{array}{ll}
1, &~\mathrm{if}~ i>j,\\
0, &~\mathrm{if}~ i\leq j.
\end{array}
\right.
\] 
The R-matrix \eqref{Rmatrix} appeared in \cite{BFHTW} as the $q\rightarrow 0$ limit of a Drinfeld twist of the R-matrix for an evaluation module of $U_q(\widehat{\mathrm{sl}}_{n+1})$, here we show that $R^{(n)}(x,y)$ satisfies the Yang-Baxter equation by a direct computation:
\begin{thm}\label{prop:YBE}
The R-matrix defined by Eq. \eqref{Rmatrix} satisfies the Yang-Baxter equation
\begin{equation}\label{YBE}
R^{(n)}_{12}(x,y) R^{(n)}_{13}(x,z) R^{(n)}_{23}(y,z) = R^{(n)}_{23}(y,z) R^{(n)}_{13}(x,z) R^{(n)}_{12}(x,y)
\end{equation}
for all $n >1$, where $R^{(n)}_{ab}$ acts on the $a$th and $b$th factor of $V \otimes V \otimes V$.
\end{thm}
\begin{proof}
A direct computation shows that Eq. \eqref{YBE} holds for $n=2$. Now suppose $R^{(n-1)}$ satisfies the Yang-Baxter equation for some $n > 2$. Let $i_{ab}$ denote the identity map from the $a$th factor to the  $b$th factor of $V \otimes V \otimes V$. Plugging \eqref{Rmatrix} into Eq. \eqref{YBE}, we find that Eq. \eqref{YBE} decomposes into the following identities:
\begin{align}
&i_{21} i_{32} = i_{31},\quad i_{13} = i_{23} i_{12},\label{ybe1} \\ 
&i_{12} i_{31} i_{23} = i_{32} i_{13} i_{21},\quad i_{21} i_{13} i_{32} = i_{23} i_{31} i_{12},\label{ybe2} \\ 
&(x-y) i_{31} i_{23}+(y-z) i_{21} = i_{21}(x-z), \quad (x-z) i_{12} = (x-y) i_{32} i_{13} + (y-z) i_{12},\label{ybe3} \\ 
&R^{(n-1)}_{12}(x,y)i_{31}R^{(n-1)}_{23}(y,z) = i_{32} R^{(n-1)}_{13}(x,z) i_{21},\label{ybe4} \\ 
&i_{12} R^{(n-1)}_{13}(x,z) i_{23} = R_{23}^{(n-1)}(y,z) i_{13} R_{12}^{(n-1)}(x,y),\label{ybe5} \\ 
&(x-z)(1+\beta y) R_{12}^{(n-1)}(x,y) i_{32} = (y-z)(1+\beta x) i_{31} i_{12} + (x-y)(1+\beta z) i_{32} R_{13}^{(n-1)}(x,z),\label{ybe6} \\ 
&(y-z)(1+\beta x)i_{21} i_{13} + (x-y)(1+\beta z) R_{13}^{(n-1)}(x,z) i_{23} = (x-z)(1+\beta y) i_{23} R_{12}^{(n-1)}(x,y), \label{ybe7} \\ 
&R^{(n-1)}_{12}(x,y) R^{(n-1)}_{13}(x,z) R^{(n-1)}_{23}(y,z) = R^{(n-1)}_{23}(y,z) R^{(n-1)}_{13}(x,z) R^{(n-1)}_{12}(x,y). \label{ybe8}
\end{align}
Eq. \eqref{ybe1}-\eqref{ybe3} are direct consequences of the definition of $i_{ab}$. Eq. \eqref{ybe4} and \eqref{ybe5} can be easily derived from the identities
\begin{equation}\label{ybe01}
\sum_{k,l} R^{(m)}_{ij,lk}(x,y) R^{(m)}_{kl,st}(y,z) = R^{(m)}_{ij,st}(x,z)
\end{equation}
and
\begin{equation}\label{ybe02}
\sum_{k,l} R^{(m)}_{ij,lk}(y,z) R^{(m)}_{kl,st}(x,y) = R^{(m)}_{ij,st}(x,z)
\end{equation}
respectively, while Eq. \eqref{ybe6} and \eqref{ybe7} can be easily derived from
\begin{equation}\label{ybe03}
(y-z)(1+\beta x) R^{(m)}(x,x) + (x-y)(1+\beta z) R^{(m)}(x,z) = (x-z)(1+\beta y) R^{(m)}(x,y),
\end{equation}
where $m\geq 2$.
Eq. \eqref{ybe01}-\eqref{ybe03} can all be proved by simple induction on $m$. Finally, Eq. \eqref{ybe8} is the inductive hypothesis. Therefore, by induction, the Yang-Baxter equation holds for all $n \geq 2$.
\end{proof}

To construct an integrable system, we take $N$ copies of $V\cong \mathbb{C}^n$ and label them by $V_i, i=1,\cdots,N$. Let $V_a \cong V$ be an auxiliary space. The monodromy matrix is defined by
\begin{equation}\label{monodromy}
T^{(n)}_a(u) = R_{aN}(u,t_N) R_{a,N-1}(u,t_{N-1}) \cdots R_{a2}(u,t_2) R_{a1}(u,t_1),
\end{equation}
where $R_{ai}$ acts on $V_a\otimes V_i$. 
The Yang-Baxter equation \eqref{YBE} leads to the following RTT relation:
\begin{equation}\label{RTT} 
R^{(n)}_{ab}(x,y) T^{(n)}_a(x) T^{(n)}_b(y) = T^{(n)}_b(y) T^{(n)}_a(x) R^{(n)}_{ab}(x,y).
\end{equation}
We define the transfer matrix to be
\begin{equation}\label{transfer}
t^{(n)}(u) = \mathrm{Tr}_{V_a} (M_n T^{(n)}_a(u)),
\end{equation}
where $M_n = \mathrm{diag}(b_0,b_1,\cdots,b_{n-1})$ acts on $V_a$, imposing a twisted periodic boundary condition. In the following, we also need
\begin{equation}\label{Mll}
M_{l} := \mathrm{diag}(b_{n-l},\cdots, b_{n-1}), \quad l=1,\cdots,n.
\end{equation}
Eq. \eqref{RTT} implies that\footnote{Notice that for any diagonal matrix $M$, we have $R^{(n)}_{ab}(x,y) M_a M_b = M_b M_a R^{(n)}_{ab}(x,y)$.}
\begin{equation}
[t(u),t(v)] = 0.
\end{equation}
Under the decomposition $V_a = \mathbb{C} \oplus \mathbb{C}^{n-1}$, the monodromy matrix can be written in the following form:
\begin{equation}\label{Tab}
T^{(n)}_a(u)= \left( \begin{array}{cccc}
A(u) & B_1(u) & \cdots & B_{n-1}(u) \\
C_1(u) & D_{11}(u) & \cdots & D_{1,n-1}(u) \\
\vdots & \vdots & \ddots & \vdots \\
C_{n-1}(u) & D_{n-1,1}(u) & \cdots & D_{n-1,n-1}(u) 
\end{array} \right)
\end{equation}
with $A(u), B_i(u), C_i(u), D_{ij}(u)$ being operators acting on the Hilbert space $\mathcal{H} = V_1 \otimes V_2 \otimes \cdots \otimes V_N$. The RTT relation \eqref{RTT} yields the following commutation relations:
\begin{align}
A(x) B_i(y) &= \frac{1}{y \ominus x} \left[ B_i(y) A(x) -  B_i(x) A(y) \right], \label{comAB} \\
D_{ij}(x) B_k(y) &= \sum_{s,l} \left[ \frac{1}{x \ominus y} B_s(y) D_{il}(x) R^{(n-1)}_{ls,jk}(x,y)\right] + \frac{1}{y \ominus x} B_j(x) D_{ik}(y), \label{comDB}\\
B_i(x) B_j(y) &= \sum_{r,s} B_r(y) B_s(x) R^{(n-1)}_{sr,ij}(x,y) \label{comBB}.
\end{align}

\section{Bethe Ansatz equations and the equivariant quantum cohomology/K-theory ring of flag varieties}\label{sec:QK}

In this section, following the method proposed in \cite{KR}, we derive the Bethe Ansatz equations of the GL$(n)$ five vertex model introduced in the previous section, and discuss their relationship with the equivariant quantum cohomology and quantum K-theory ring of flag varieties.

Let $V = \mathbb{C}^n$ be spanned by $v_0,v_1,\cdots,v_{n-1}$, where $v_i := \ket{i}$ is the $(i+1)$th unit vector of $\mathbb{C}^n$. Define
\[
\ket{n_1,n_2,\cdots,n_N} :=  v_{n_1} \otimes v_{n_2} \otimes \cdots \otimes v_{n_N},
\]
then a basis of the Hilbert space $\mathcal{H} = V_1 \otimes V_2 \otimes \cdots V_N$ can be chosen to be
\begin{equation}\label{basis}
\{ \ket{n_1,n_2,\cdots,n_N} |~ 0 \leq n_i \leq n-1 ~\}.
\end{equation}
We call the basis \eqref{basis} the natural basis of the Hilbert space $\mathcal{H}$.
The pseudo vacuum is taken to be
\begin{equation}\label{ground}
\ket{\Omega^{(0)}} = \ket{0,0,\cdots ,0} \in  V_1 \otimes V_2 \otimes \cdots \otimes V_N.
\end{equation}
From Eq. \eqref{Rmatrix} and \eqref{Tab}, it is easy to compute
\begin{equation}
A(u) \ket{\Omega^{(0)}} = \ket{\Omega^{(0)}},\quad D_{ij}(u) \ket{\Omega^{(0)}} = \prod_{k=1}^N (u \ominus t_k) \delta_{ij} \ket{\Omega^{(0)}}.
\end{equation}
The Bethe Ansatz state is taken to be of the form
\begin{equation}\label{BEstate}
\ket{\psi^{(0)} } = \sum_{i_1,\cdots,i_{k_1}} \psi_{i_1\cdots i_{k_1}}^{(1)} B_{i_1}(\sigma_1^{(1)}) \cdots B_{i_{k_1}}(\sigma_{k_1}^{(1)}) \ket{\Omega^{(0)}},
\end{equation}
where $\sigma_{\alpha_1}^{(1)}$, $\alpha_1 = 1,\cdots,k_1$, are parameters to be determined, and the expansion coefficients $\psi_{i_1\cdots i_{k_1}}^{(1)}$ depend on $n-2$ sets of parameters
\[
\sigma^{(2)}_{\alpha_2} , \cdots, \sigma^{(n-1)}_{\alpha_{n-1}}
\]
with $\alpha_s = 1,2,\cdots, k_s$, where $\sigma^{(s)}_a \neq \sigma^{(s)}_b$ for $a \neq b$ and $1 \leq s \leq n-1$. Now we can prove the following
\begin{thm}\label{BAE}
If the parameters $\sigma^{(1)}_{\alpha_1} , \cdots, \sigma^{(n-1)}_{\alpha_{n-1}}$ satisfy the following Bethe Ansatz equations:
\begin{align}
&\prod_{a=1}^{k_1}(1+\beta(\sigma^{(1)}_a \ominus \sigma^{(1)}_{\alpha_1})) \prod_{i=1}^N (\sigma^{(1)}_{\alpha_1} \ominus t_i) = (-1)^{k_1-1} q_1 \prod_{b=1}^{k_2} (\sigma^{(2)}_b \ominus \sigma^{(1)}_{\alpha_1}),\label{BAE1}\\
&\prod_{a=1}^{k_m}(1+\beta(\sigma^{(m)}_a \ominus \sigma^{(m)}_{\alpha_m})) \prod_{i=1}^{k_{m-1}} (\sigma^{(m)}_{\alpha_m} \ominus \sigma^{(m-1)}_i) = (-1)^{k_m-1} q_m \prod_{b=1}^{k_{m+1}} (\sigma^{(m+1)}_b \ominus \sigma^{(m)}_{\alpha_m}),\label{BAE2}\\
&\prod_{a=1}^{k_{n-1}}(1+\beta(\sigma^{(n-1)}_a \ominus \sigma^{(n-1)}_{\alpha_{n-1}})) \prod_{i=1}^{k_{n-2}} (\sigma^{(n-1)}_{\alpha_{n-1}} \ominus \sigma^{(n-2)}_i) = (-1)^{k_{n-1}-1} q_{n-1},\label{BAE3}
\end{align}
where $m=2,\cdots,n-2,\quad \alpha_s=1,\cdots,k_s, \quad q_s := b_{s-1}/b_s,~ s=1,\cdots,n-1$,
then the state \eqref{BEstate} is a common eigenstate of the transfer matrices $t(u)$ defined by \eqref{transfer}.
\end{thm}
\begin{proof}
From the commutation relations \eqref{comAB} and \eqref{comDB}, one can compute
\begin{equation}\label{ADpsi}
\begin{aligned}
&A(u) \ket{\psi^{(0)}} = \prod_{a=1}^{k_1} (\sigma_a^{(1)} \ominus u)^{-1} \ket{\psi^{(0)}} + \ket{\phi_A},\\
&D_{ij}(u) \ket{\psi^{(0)}} = \frac{\prod_{l=1}^N (u \ominus t_l)}{\prod_{a=1}^{k_1} (u \ominus \sigma^{(1)}_a)}\\ 
&\times \sum_{i_1,\cdots,i_{k_1}}\sum_{j_1,\cdots,j_{k_1}} \left\{ [T_{ij}^{(n-1)}(u;\sigma^{(1)})]_{i_1\cdots i_{k_1},j_1\cdots j_{k_1}} \psi_{j_1\cdots j_{k_1}}^{(1)} B_{i_1}(\sigma_1^{(1)}) \cdots B_{i_{k_1}}(\sigma_{k_1}^{(1)}) \right\} \ket{\Omega^{(0)}}+\ket{\phi_D},
\end{aligned}
\end{equation}
where $\ket{\phi_A}$ and $\ket{\phi_D}$ contain all the unwanted terms that are not in the subspace spanned by states of the form $B_{i_1}(\sigma_1^{(1)}) \cdots B_{i_{k_1}}(\sigma_{k_1}^{(1)}) \ket{\Omega^{(0)}}$, and
\begin{equation}\label{Tn-1}
\begin{aligned}
&[T_{ij}^{(n-1)}(u;\sigma^{(1)})]_{i_1\cdots i_{k_1},j_1\cdots j_{k_1}} \\&= \sum_{\mu_1 \cdots \mu_{k_1-1}} R^{(n-1)}_{ii_{k_1},\mu_{k_1-1}j_{k_1}}(u,\sigma^{(1)}_{k_1})
\cdots R^{(n-1)}_{\mu_2 i_2,\mu_1 j_2}(u,\sigma^{(1)}_2) R^{(n-1)}_{\mu_1 i_1,j j_1}(u,\sigma^{(1)}_1).
\end{aligned}
\end{equation}
Eq. \eqref{Tn-1} suggests that $T^{(n-1)}$ can be viewed as the monodromy matrix of the GL$(n-1)$ asymmetric five vertex model with $k_1$ sites and R-matrix $R^{(n-1)}$, in which $\sigma^{(1)}_a, a=1,\cdots,k_1$, are the equivariant parameters. Correspondingly, the transfer matrix can be defined as
\[
t^{(n-1)}(u;\sigma^{(1)}) = \mathrm{Tr}(M_{n-1} T^{(n-1)}(u;\sigma^{(1)})),
\]
where $M_l$ is defined by \eqref{Mll}.
Let
\[
\psi_{i_1\cdots i_{k_1}}^{(1)} = \braket{i_1 i_2 \cdots i_{k_1} | \psi^{(1)}}, \quad i_s = 1,\cdots, n-1.
\] 
If $\ket{\psi^{(1)}}$ is an eigenstate of $t^{(n-1)}(u;\sigma^{(1)})$ with eigenvalue $\Lambda^{(1)}(u)$, i.e.
\begin{equation}\label{eig_n-1}
t^{(n-1)}(u;\sigma^{(1)}) \ket{\psi^{(1)}} = \Lambda^{(1)}(u) \ket{\psi^{(1)}},
\end{equation} 
then, as a consequence of Eq. \eqref{ADpsi}
\[
\mathrm{Tr}( M_{n-1} D(u)) \ket{\psi^{(0)}} = \frac{\prod_{i=1}^N (u \ominus t_i)}{\prod_{a=1}^{k_1} (u \ominus \sigma^{(1)}_a)} \Lambda^{(1)}(u) \ket{\psi^{(0)}} + \ket{\tilde{\phi}_D},
\]
where $\ket{\tilde{\phi}_D}$ is the sum of the unwanted terms.
Then
\[
t^{(n)}(u) \ket{\psi^{(0)}} = \left[b_0 A(u) + \mathrm{Tr}( M_{n-1} D(u) )\right] \ket{\psi^{(0)}} =
\Lambda^{(0)}(u) \ket{\psi^{(0)}} + \ket{\phi^{(n)}},
\]
where
\begin{equation}\label{eigv_0}
\Lambda^{(0)}(u) = b_0 \prod_{a=1}^{k_1}(\sigma_a^{(1)} \ominus u)^{-1} + \frac{\prod_{i=1}^N (u \ominus t_i)}{\prod_{a=1}^{k_1} (u \ominus \sigma^{(1)}_a)} \Lambda^{(1)}(u),
\end{equation}
and $\ket{\phi^{(n)}}$ is the sum of the unwanted terms. For $\ket{\psi^{(0)}}$ to be an eigenstate of $t^{(n)}(u)$ for all $u$, the unwanted terms must vanish. As in the usual algebraic Bethe Ansatz method, vanishing of the unwanted terms $\ket{\phi^{(n)}}$ amounts to $\Lambda^{(0)}(u)$ being holomorphic, i.e.
\begin{equation}\label{res}
\mathrm{Res} (\Lambda^{(0)}; \sigma^{(1)}_a) = 0
\end{equation}
for all $a = 1,\cdots, k_1$. Eq. \eqref{res} gives us the equations for $\sigma_a^{(1)}$:
\begin{equation}\label{BAE0}
\prod_{\alpha=1}^{k_1}(1+\beta \sigma_\alpha^{(1)}) \prod_{i=1}^N(\sigma_a^{(1)}\ominus t_i) \Lambda^{(1)}(\sigma_a^{(1)}) = (-1)^{k_1-1} b_0 (1+\beta \sigma_a^{(1)})^{k_1}, \quad a=1,\cdots,k_1.
\end{equation}
Thus we have reduced the eigenvalue problem for the GL$(n)$ system to the eigenvalue problem \eqref{eig_n-1} for a GL$(n-1)$ system. Clearly, we can continue to reduce the rank in exactly the same way. At each step, a set of parameters $\sigma^{(m)}_{\alpha_m}, \alpha_m=1,\cdots,k_m$, has to be introduced to construct the Bethe Ansatz state
\begin{equation}\label{psis}
\ket{\psi^{(m-1)}} = \sum_{i_1,\cdots,i_{k_m}} \psi^{(m)}_{i_1,\cdots,i_{k_m}} B^{(m-1)}_{i_1}(\sigma^{(m)}_1) \cdots B^{(m-1)}_{i_{k_m}}(\sigma^{(m)}_{k_m}) \ket{\Omega^{(m-1)}},
\end{equation}
where
\begin{equation}\label{psim}
 \psi^{(m)}_{i_1,\cdots,i_{k_m}} = \braket{i_1,\cdots,i_{k_m}|\psi^{(m)}},
\end{equation}
and $B^{(m-1)}_i$'s are the $B_i$ operators defined as matrix elements of the monodromy matrix in the corresponding step as in \eqref{Tab} and $\sigma^{(m-1)}_a, a=1,\cdots,k_{m-1}$, serve as the equivariant parameters.
As in the $m=1$ case, for \eqref{psis} to be an eigenstate, the following equations have to be satisfied:
\begin{equation}\label{BAEm}
\prod_{b=1}^{k_m}(1+\beta \sigma_b^{(m)}) \prod_{a=1}^{k_{m-1}}(\sigma_{\alpha_m}^{(m)}\ominus \sigma_a^{(m-1)}) \Lambda^{(m)}(\sigma_{\alpha_m}^{(m)}) = (-1)^{k_m-1} b_{m-1} (1+\beta \sigma_{\alpha_m}^{(m)})^{k_m}, \quad \alpha_m=1,\cdots,k_m,
\end{equation}
accordingly, the eigenvalues satisfy
\begin{equation}\label{eigv_m-1}
\Lambda^{(m-1)}(u) = b_{m-1} \prod_{\alpha=1}^{k_m}(\sigma_{\alpha}^{(m)} \ominus u)^{-1} + \frac{\prod_{a=1}^{k_{m-1}} (u \ominus \sigma_a^{(m-1)})}{\prod_{b=1}^{k_m} (u \ominus \sigma^{(m)}_b)} \Lambda^{(m)}(u)
\end{equation}
for $m=2,\cdots,n-2$.

This procedure continues until we get a GL$(2)$ system with R-matrix given by Eq. \eqref{R2}. This GL$(2)$ system can be solved exactly by the algebraic Bethe Ansatz method as for other spin chain models, leading to the Bethe Ansatz equations Eq. \eqref{BAE3} and
\begin{equation}\label{eigv_n-2}
\Lambda^{(n-2)}(u) =b_{n-2} \prod_{\alpha=1}^{k_{n-1}}(\sigma_{\alpha}^{(n-1)} \ominus u)^{-1} + b_{n-1}\frac{\prod_{a=1}^{k_{n-2}} (u \ominus \sigma_a^{(n-2)})}{\prod_{b=1}^{k_{n-1}} (u \ominus \sigma^{(n-1)}_b)}.
\end{equation}
Eq. \eqref{eigv_n-2} then allows us to solve all the eigenvalues $\Lambda^{(m)}, m=0,1,\cdots,n-3$, from Eq. \eqref{eigv_m-1} and \eqref{eigv_0} recursively. Plugging these eigenvalues into Eq. \eqref{BAEm} and \eqref{BAE0}, we get the equations Eq. \eqref{BAE2} and \eqref{BAE1}.
\end{proof}

Now let us restrict to the cases in which $k_{n-1} < \cdots < k_2 < k_1 < N$ and compare Eq. \eqref{BAE1}-\eqref{BAE3} with the ring relations of the equivariant quantum cohomology and the equivariant quantum K-theory ring of flag varieties. The reader may refer to the appendix and references therein for a review of the quantum cohomology and quantum K-theory ring of flag varieties.  

When $\beta=0$, we can identify $\sigma_a^{(n-i)}$ with the Chern roots of $\mathcal{S}_i$, the $i$th tautological bundle of the flag variety $\mathrm{Fl}(k_{n-1},\cdots,k_2,k_1;N)$, i.e. the $j$th elementary symmetric polynomial $e_j(\sigma_1^{(n-i)},\cdots,\sigma_{k_{n-i}}^{(n-i)})$ is identified with the $j$th Chern class of $\mathcal{S}_{i}$:
\[
e_j(\sigma_1^{(n-i)},\cdots,\sigma_{k_{n-i}}^{(n-i)}) = c_j(\mathcal{S}_{i}),\quad j=1,\cdots,k_{n-i}.
\]
Under this identification, it was checked in \cite{Guo, GMSXZZ} that Eq. \eqref{BAE1}-\eqref{BAE3} specialized to $\beta=0$ are exactly the ring relations of the equivariant quantum cohomology of $\mathrm{Fl}(k_{n-1},\cdots,k_2,k_1;N)$ with the equivariant parameters given by $(t_1,\cdots,t_N)$.

When $\beta=-1$, we can identify $1-\sigma_a^{(n-i)}$ with the Chern roots of $\mathcal{S}_i$, i.e. the $j$th elementary symmetric polynomial $e_j(1-\sigma_1^{(i)},\cdots,1-\sigma_{k_{n-i}}^{(i)})$ is identified with $\wedge^j \mathcal{S}_i$ in the Grothendieck group of the flag variety $\mathrm{Fl}(k_{n-1},\cdots,k_2,k_1;N)$:
\[
e_j(1-\sigma_1^{(n-i)},\cdots,1-\sigma_{k_{n-i}}^{(n-i)}) = \wedge^j \mathcal{S}_i,\quad j=1,\cdots,k_{n-i}.
\]
It was argued in \cite{GMSXZZ} that Eq. \eqref{BAE1}-\eqref{BAE3} specialized to $\beta=-1$ are the ring relations of the equivariant quantum K-theory ring of $\mathrm{Fl}(k_{n-1},\cdots,k_2,k_1;N)$ with the equivariant parameters given by $(1-t_1,\cdots,1-t_N)$. 

Indeed, as in \cite{GMSXZZ}, one can use Vieta's formula\footnote{Vieta's formula: If $r_1,r_2,\cdots,r_n$ are roots of the polynomial $a_n x^n + a_{n-1} x^{n-1} + \cdots + a_1 x+a_0$, then $e_l(r_1,\cdots,r_n) = (-1)^l a_{n-l}/a_n$.} to convert Eq. \eqref{BAE1}-\eqref{BAE3} into the ring relations of the equivariant quantum cohomology ring ($\beta=0$) or the equivariant K-theory ring ($\beta=-1$) of the flag variety $\mathrm{Fl}(k_{n-1},\cdots,k_2,k_1;N)$. 

When $\beta = 0$, Eq. \eqref{BAE1}-\eqref{BAE3} reduces to
\begin{equation}\label{BAEQC}
\prod_{\alpha=1}^{k_{m-1}}(\sigma^{(m)}_a - \sigma^{(m-1)}_\alpha) = (-1)^{k_m-1}q_m \prod_{\gamma=1}^{k_{m+1}} (\sigma^{(m+1)}_\gamma - \sigma^{(m)}_a)
\end{equation}
for $a=1,\cdots,k_m, m=1,2,\cdots,n-1$, where we have defined $\sigma^{(0)}_i=t_i, k_0=N, k_n=0$. Eq. \eqref{BAEQC} tells us that $\sigma^{(m)}_a, a=1,\cdots,k_m$, are roots of the polynomial
\begin{equation}\label{polyQC}
\sum_{i=0}^{k_{m-1}} (-1)^i x^{k_{m-1}-i} e_i(\sigma^{(m-1)}) + (-1)^{k_m-k_{m+1}}q_m \sum_{j=0}^{k_{m+1}} (-1)^j x^{k_{m+1}-j} e_j(\sigma^{(m+1)}),
\end{equation}
where $e_i(\sigma^{(s)})$ is the $i$th elementary symmetric polynomial in the variables $\sigma^{(s)}_a$, $a=1,\cdots,k_s$.
Notice that the coefficient of $x^{k_{m-1}-l}$ in \eqref{polyQC} is\footnote{We set $e_l(x_1,\cdots,x_s)=0$ for $l<0$ or $l>s$.} 
\[
(-1)^l e_l(\sigma^{(m-1)})+(-1)^{k_{m-1}-k_m+l}q_m e_{k_{m+1}-k_{m-1}+l}(\sigma^{(m+1)}).
\]
Assume the other $k_{m-1} - k_m$ roots of \eqref{polyQC} are $\eta^{(m)}_b, b=1,\cdots,k_{m-1}-k_m$, then Vieta's formula yields
\begin{equation}\label{relQC}
\sum_{i=0}^l e_i(\sigma^{(m)}) e_{l-i}(\eta^{(m)}) = e_l(\sigma^{(m-1)}) + (-1)^{k_{m-1}-k_m}q_m e_{k_{m+1}-k_{m-1}+l}(\sigma^{(m+1)}).
\end{equation}
If we interpret $\sigma^{(m)}_a, a=1,\cdots,k_m,$ as Chern roots of $\mathcal{S}_{n-m}$ and $\eta^{(m)}_b, b=1,\cdots,k_{m-1}-k_m,$ as Chern roots of $\mathcal{S}_{n-m+1}/\mathcal{S}_{n-m}$, i.e. $e_i(\sigma^{(m)}) = c_i(\mathcal{S}_{n-m})$, $e_i(\eta^{(m)}) = c_i(\mathcal{S}_{n-m+1}/\mathcal{S}_{n-m})$, then Eq. \eqref{relQC} becomes\footnote{$\mathcal{S}_0$ is the trivial vector bundle of rank zero, so $c({\mathcal{S}_0}) = 1$. $\mathcal{S}_n$ is the trivial vector bundle of rank $N$, but since it carries the $T$-action, we have $c(\mathcal{S}_n) = \sum_{i=0}^N e_i(t_1,\cdots,t_N)$.}
\begin{equation}\label{whitneyQC}
c(\mathcal{S}_i) \cdot c(\mathcal{S}_{i+1}/\mathcal{S}_i) = c(\mathcal{S}_{i+1}) + (-1)^{k_{n-i-1}-k_{n-i}} q_{n-i} c(\mathcal{S}_{i-1}),
\end{equation}
for $i=1,2,\cdots,n-1$,
which are the quantum Whitney relations of the equivariant quantum cohomology ring $QH^*_T(\mathrm{Fl}(k_{n-1},\cdots,k_2,k_1;N))$ and are equivalent to the ring relations proved in \cite{AS95, K95}.

When $\beta \neq 0$, let $X^{(m)}_a = 1 + \beta \sigma^{(m)}_a$, then Eq. \eqref{BAE1}-\eqref{BAE3} can be written as
\begin{equation}\label{BAEQK}
\begin{aligned}
&\prod_{\alpha=1}^{k_{m-1}}\left( X^{(m)}_a - X^{(m-1)}_\alpha \right) \prod_{l=1}^{k_m}\left( X^{(m)}_l \right)\\ 
= &(-1)^{k_m-1} q_m \beta^{k_{m-1}-k_{m+1}} \prod_{s=1}^{k_{m-1}}\left( X^{(m-1)}_s \right)\prod_{\gamma=1}^{k_{m+1}} \left( X^{(m+1)}_\gamma -X^{(m)}_a \right) \left( X^{(m)}_a \right)^{k_m-k_{m+1}}
\end{aligned}
\end{equation}
for $a=1,\cdots,k_m, m=1,2,\cdots,n-1$, where $X^{(0)}_i = 1+ \beta t_i, k_0=N, k_n = 0$. Eq. \eqref{BAEQK} implies that $X^{(m)}_a, a=1,\cdots,k_m,$ are roots of the polynomial $\sum_{l=0}^{k_{m-1}} a_l x^{l}$, where
\[
\begin{aligned}
a_{k_{m-1}-l} = &(-1)^l e_l(X^{(m-1)}) e_{k_m}(X^{(m)})\\&+ (-1)^{k_{m-1}-k_{m+1}+l} q_m \beta^{k_{m-1}-k_{m+1}} e_{k_{m-1}}(X^{(m-1)}) e_{l-k_{m-1}+k_m}(X^{(m+1)}).
\end{aligned}
\]
Assume the roots of $\sum_{l=0}^{k_{m-1}} a_l x^{l}$ are $X^{(m)}_a, a=1,\cdots,k_m,$ and $Y^{(m)}_b, b=1,\cdots,k_{m-1}-k_m$, then Vieta's formula yields
\begin{equation}\label{relQK}
\begin{aligned}
&e_{k_m}(X^{(m)}) \sum_{i+j=l} e_i(X^{(m)}) e_j(Y^{(m)})\\ 
= &e_{k_m}(X^{(m)}) e_l(X^{(m-1)}) + q_m (-\beta)^{k_{m-1}-k_{m+1}} e_{k_{m-1}}(X^{(m-1)}) e_{l-k_{m-1}+k_m}(X^{(m+1)}).
\end{aligned}
\end{equation}
Taking $l=k_{m-1}$ in Eq. \eqref{relQK}, we get
\begin{equation}\label{subQK}
e_{k_{m-1}}(X^{(m-1)}) = e_{k_m}(X^{(m)}) e_{k_{m-1}-k_m}(Y^{(m)}).
\end{equation}
Eq. \eqref{relQK} and \eqref{subQK} together give us the relation
\begin{equation}\label{relQKc}
\sum_{i+j=l} e_i(X^{(m)}) e_j(Y^{(m)}) =  e_l(X^{(m-1)}) + q_m (-\beta)^{k_{m-1}-k_{m+1}}  e_{l-k_{m-1}+k_m}(X^{(m+1)}) e_{k_{m-1}-k_m}(Y^{(m)}).
\end{equation}
In the case $\beta = -1$, \eqref{relQKc} was interpreted as the quantum Whitney relation of the equivariant quantum K-theory ring of the flag variety in \cite{GMSXZZ}. The idea is to identify $e_i(X^{(m)})$ with $\wedge^i \mathcal{S}_{n-m}$ and make the following identification
\[
e_l(Y^{(m)}) = \left\{ \begin{array}{ll}
\wedge^l(\mathcal{S}_{n-m+1}/\mathcal{S}_{n-m}), & l< k_{m-1}-k_m,\\
\frac{1}{1-q_m} \det (\mathcal{S}_{n-m+1}/\mathcal{S}_{n-m}), & l= k_{m-1}-k_m.
\end{array} \right.
\]
Then Eq. \eqref{relQKc} for $\beta = -1$ can be written in terms of the $\lambda_y$ classes as\footnote{For a vector bundle $\mathcal{E}$, the $\lambda_y$ class is defined as $\lambda_y(\mathcal{E}) = \sum_{i=0}^{\mathrm{rank}(\mathcal{E})}y^i \wedge^i \mathcal{E}$ in the Grothendieck group.}
\begin{equation}\label{whitneyQK}
\lambda_y(\mathcal{S}_i) * \lambda_y(\mathcal{S}_{i+1}/\mathcal{S}_i) = \lambda_y(\mathcal{S}_{i+1}) + \frac{q_{n-i}}{1-q_{n-i}} y^{k_{n-i-1}-k_{n-i}}\left( \lambda_y(\mathcal{S}_{i-1}) -\lambda_y(\mathcal{S}_i) \right) * \det(\mathcal{S}_{i+1}/\mathcal{S}_i)
\end{equation}
for $i=1,2,\cdots,n-1$,
which are the quantum Whitney relations of the equivariant quantum K-theory ring $QK_T(\mathrm{Fl}(k_{n-1},\cdots,k_2,k_1;N))$ proposed in \cite{GMSXZZ} and proved in \cite{Huq-Kuruvilla:2024tsg}.

\begin{remark} 
Notice that in \cite{GMSXZZ}, the analogue of the Bethe Ansatz equations were derived from the 3d $\mathcal{N}=2$ gauged linear sigma model (GLSM) compactified on a circle. The GLSMs for the flag varieties are quiver gauge theories, whose quiver diagrams have the following form in our convention \cite{DS}
\[
\xymatrix@C=3pc{
*++[o][F]{k_{n-1}} & *++[o][F]{k_{n-2}} \ar[l] & \cdots \ar[l] & *++[o][F]{~k_{1~}} \ar[l] & *++[F]{N} \ar[l]}
\]
where circles represent $U(k_i)$ gauge groups, square represents $U(N)$ global symmetry and arrows represent chiral fields in bifundamental representations.
The Higgs branch of the GLSM is a nonlinear sigma model with target space being the flag variety $\mathrm{Fl}(k_{n-1},\cdots,k_2,k_1;N)$, the equations for the vacuum states on the Coulomb branch are exactly the Bethe Ansatz equations we obtained above in the case of $\beta = -1$. 
\end{remark}

\begin{remark}
Theorem \ref{BAE} shows that the Bethe Ansatz equations of the quantum integrable model defined by the R-matrix \eqref{Rmatrix} and monodromy matrix \eqref{monodromy} reduce to the vacuum equations on the Coulomb branch of the 2d and 3d GLSM for flag varieties when $\beta = 0$ and $\beta=-1$ respectively. Therefore, we have established a Bethe/Gauge correspondence in the sense of \cite{NS1, NS2}.
\end{remark}

\section{Relationship with the double $\beta$-Grothendieck polynomials}\label{sec:Groth}

In this section, we consider the case of complete flag variety, i.e. $n=N$ and $k_{N-i} = i$ for $i = 1,\cdots,N-1$. We will show that the Bethe Ansatz states generate the double $\beta$-Grothendieck polynomials (Theorem \ref{thm1} and \ref{thm2}). The double $\beta$-Grothendieck polynomials $\mathcal{G}^{(\beta)}_w(\bm{x};\bm{y})$ for $w \in S_N$ and $\bm{x} = (x_1,\cdots,x_N), \bm{y} = (y_1,\cdots,y_N)$ are defined in the appendix, see Eq. \eqref{didi}-\eqref{otherGroth}.

In addition to Eq. \eqref{ominus}, let us also define
\begin{equation}\label{ot}
\ominus t := \frac{-t}{1+\beta t} 
\end{equation}
and $\ominus \bm{t} := (\ominus t_1,\cdots, \ominus t_N)$ for $\bm{t} = (t_1,\cdots,t_N)$. We use $\omega_N$ to denote the permutation of $N$ elements with maximal length, and $l(w)$ denotes the length of the permutation $w$.

From Eq. \eqref{Rijkl}, \eqref{monodromy} and the fact that $B_k(u) = [{T^{(n)}_a(u)}]_{0,k}$ (Eq. \eqref{Tab}), it is easy to compute
\begin{equation}\label{exdR}
\begin{aligned}
&B_k(u) \ket{n_1,n_2,\cdots,n_N} = \\
&= \sum_{m_1=0}^{N-1} \cdots \sum_{m_{N-1}=0}^{N-1}{R_{aN}(u,t_N)}_{0,m_{N-1}} 
\cdots {R_{a2}(u,t_2)}_{m_2,m_1} {R_{a1}(u,t_1)}_{m_1,k} \ket{n_1,n_2,\cdots,n_N} \\
&= \sum_{m_1=0}^{N-1} \cdots \sum_{m_{N-1}=0}^{N-1} \bigotimes_{i=1}^N\left[ \delta_{m_i,n_i}(1+\delta_{m_i>m_{i-1}}\beta(u\ominus t_i) ) \ket{m_{i-1}}+ \delta_{m_i > n_i} \delta_{m_i,m_{i-1}} (u\ominus t_i)  \ket{n_i} \right],
\end{aligned}
\end{equation}
where $m_0 = k, m_N=0$. The state \eqref{exdR} vanishes unless $m_i \geq n_i$ for all $1 \leq i \leq N$, therefore it vanishes when $n_N \neq 0$. In addition, $m_i = m_{i-1}$ if $m_i > n_i$.

\begin{figure}
\centering
\includegraphics[width=\linewidth]{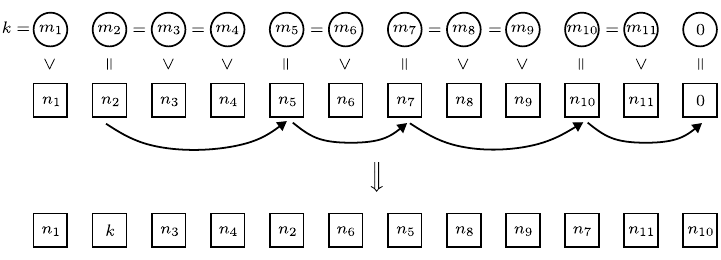}
\caption{The left action of $B_k(u)$ on $\ket{n_1,n_2,\cdots,n_N}$ for $N=12$ and a specific choice of $(m_1,\cdots,m_{N-1})$ in the expansion \eqref{exdR}. Any $n_i$ that satisfies $n_i < m_i$ remains unchanged. The curved arrow indicates that every $n_i$ satisfying $n_i=m_i$ is replaced by the preceding $n_j$ that satisfies $n_j=m_j$, while the first $n_i$ satisfying $n_i=m_i$ is replaced by $k$.}
\label{fig:Laction}
\end{figure}
The left action of $B_k(u)$ on $\ket{n_1,n_2,\cdots,n_N}$ can then be depicted as Fig.\ref{fig:Laction} (for a specific $(m_1,\cdots,m_{N-1})$). We see for any $\bm{m}=(m_1,\cdots,m_{N})$ in the expansion \eqref{exdR}, the effect of the operator $B_k(u)$ is to annihilate a $\ket{0}$ state, create a $\ket{k}$ state, and permute $n_i$ among the sites according to $\bm{m}$. If we write the expansion \eqref{exdR} as $B_k(u) \ket{\bm{n}} = \sum_{\bm{m}} C_{\bm{m}} \ket{\psi_{\bm{m}}}$, then for any $n_i < m_i$, there is a factor $(u \ominus t_i)$ in $C_{\bm{m}}$, and for any $m_i > m_{i-1}$, there is a factor $(1+\beta(u\ominus t_i))$ in $C_{\bm{m}}$. 

For $w \in S_N$, let us define
\begin{equation}\label{ketstate}
\ket{w} := \ket{w(012\cdots N-1)},
\end{equation}
then, from \eqref{exdR} and the discussion above, we have the following
\begin{lem}\label{l0}
For any $\pi \in S_N$, we have the expansion
\[
B_{\pi(1)}(u_1) B_{\pi(2)}(u_2) \cdots B_{\pi(N-1)}(u_{N-1}) \ket{\Omega^{(0)}} = \sum_{w \in S_{N}} C^{\pi}_w (u_1,\cdots,u_{N-1}) \ket{w},
\]
where $C^{\pi}_w \in \mathbb{C}[u_1,\cdots,u_{N-1},t_1,\cdots,t_N]$ and $\ket{\Omega^{(0)}}$ is defined by Eq. \eqref{ground}.
\end{lem}
Similarly, the right action of $B_k(u)$ on the dual state $\bra{n_1,n_2,\cdots,n_N}$ is
\begin{equation}\label{exdL}
\begin{aligned}
&\bra{n_1,n_2,\cdots,n_N} B_k(u)\\ 
&= \sum_{m_2=0}^{N-1} \cdots \sum_{m_N=0}^{N-1} \bra{n_1,n_2,\cdots,n_N} 
{R_{aN}(u,t_N)}_{0,m_N} 
\cdots {R_{a2}(u,t_2)}_{m_3,m_2} {R_{a1}(u,t_1)}_{m_2,k}\\
&=\sum_{m_2=0}^{N-1} \cdots \sum_{m_N=0}^{N-1} \bigotimes_{i=1}^N\left[ \delta_{m_i,n_i}(1+\delta_{m_i<m_{i+1}}\beta(u\ominus t_i) ) \bra{m_{i+1}} + \delta_{m_i>n_i} \delta_{m_i,m_{i+1}} (u\ominus t_i)   \bra{n_i} \right],
\end{aligned} 
\end{equation}
where $m_1=k, m_{N+1} = 0$. The state \eqref{exdL} vanishes unless (i) $m_i \geq n_i$ for all $1 \leq i \leq N$; (ii) there is some $n_a$ such that $n_a = k = m_a$ and $m_j > n_j$ for all $j<a$. Moreover, $m_i = m_{i+1}$ if $m_i > n_i$. 

The right action of $B_k(u)$ on $\bra{n_1,n_2,\cdots,n_N}$ can therefore be depicted as Fig.\ref{fig:Raction} (for a specific choice of $(m_2,\cdots,m_N)$). We see for any $\bm{m}$ in the expansion \eqref{exdL}, the effect of the operator $B_k(u)$ is to annihilate a $\bra{k}$ state, create a $\bra{0}$ state on the $N$th site, and permute $n_i$ among the sites according to $\bm{m}$. If we write the expansion \eqref{exdL} as $\bra{\bm{n}}B_k(u) = \sum_{\bm{m}} \bra{\psi_{\bm{m}}} C_{\bm{m}}$, then for any $n_i < m_i$, there is a factor $(u \ominus t_i)$ in $C_{\bm{m}}$, and there is a factor $(1+\beta(u\ominus t_i))$ in $C_{\bm{m}}$ for any $m_i < m_{i+1}$. 
\begin{figure}
\centering
\includegraphics[width=\linewidth]{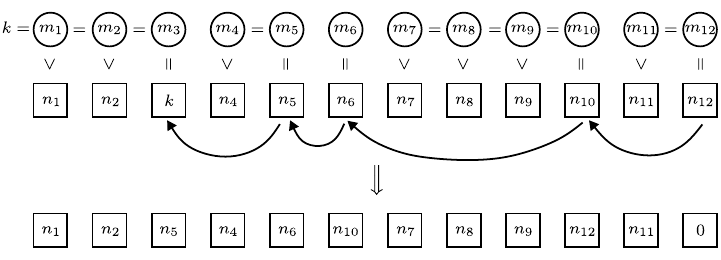}
\caption{The right action of $B_k(u)$ on $\bra{n_1,n_2,\cdots,n_N}$ for $N=12$ and a specific choice of $(m_2,\cdots,m_N)$ in the expansion \eqref{exdL}. Any $n_i$ that satisfies $n_i < m_i$ remains unchanged. The curved arrow indicates that each $n_i$ satisfying $n_i=m_i$ is replaced by the next $n_j$ that satisfies $n_j=m_j$, while the last $n_i$ satisfying $n_i=m_i$ is replaced by 0.}
\label{fig:Raction}
\end{figure}

For $w \in S_N$, define
\begin{equation}
\bra{w} := \bra{w(012\cdots N-1)},
\end{equation}
then, from the expansion Eq. \eqref{exdL}, one can prove
\begin{lem}\label{lemma1}
If $(n_1,n_2,\cdots,n_{N-1})$ satisfies: (i) $0 \leq n_i < k$; (ii) $n_{l+1} \geq n_{l+2} \geq \cdots \geq n_{N-1}$ for some $l \geq 0$, then
\[
\bra{n_{\pi(1)}, n_{\pi(2)},\cdots, n_{\pi(l)}, k, n_{l+1},\cdots, n_{N-1} } B_k(u) = \prod_{i=1}^l(u \ominus t_i) \bra{n_{\pi(1)},\cdots, n_{\pi(l)}, n_{l+1},\cdots,n_{N-1},0}
\]
for any $\pi \in S_l$.
\end{lem}
\begin{proof}
Given the constraints on $n_i$, the only possible $(m_2,m_3,\cdots,m_N)$ in the expansion \eqref{exdL} is
\[
m_i = \left\{ \begin{array}{ll}
k, & 2 \leq i \leq l+1,\\
n_i, & l+2 \leq i \leq N.
\end{array}\right.
\]
Since $n_{\pi(i)}< m_i$ for $1 \leq i \leq l$, we have the factor $\prod_{i=1}^l(u \ominus t_i)$ from \eqref{exdL}.
\end{proof}
Let us define
\begin{equation}\label{rho}
\rho_k := s_1 s_2 \cdots s_k \in S_N,
\end{equation}
where $s_i$ is the simple transposition of the $i$th and $(i+1)$th elements, and
$\rho_0$ is defined to be the identity permutation. Let $\partial^{\beta}_i$ be the $\beta$-divided difference operator defined by Eq. \eqref{didi}:
\[
\partial^{\beta}_i f(x_1,\cdots,x_N) := \frac{(1+\beta x_{i+1})f(x_1,\cdots,x_N)-(1+\beta x_i)f(x_1,\cdots,x_{i+1},x_i,\cdots,x_N)}{x_i-x_{i+1}}.
\]
Generally, for $w = s_{i_1} s_{i_2} \cdots s_{i_{l(w)}}$, we have the definition\footnote{In the following, if a polynomial $f(\bm{x};\bm{y})$ depends on two sets of variables $\bm{x}$ and $\bm{y}$ as in the case of double $\beta$-Grothendieck polynomials, $\partial^{\beta}_w$ in $\partial^{\beta}_w f(\bm{x};\bm{y})$ is understood to act on the first set of variables, see Eq. \eqref{didi}. The same convention applies to the operator $\bar{\partial}^{\beta}_w$ defined by Eq. \eqref{dbars}.}
\begin{equation}\label{ds}
\partial^{\beta}_w := \partial^{\beta}_{i_1} \partial^{\beta}_{i_2} \cdots \partial^{\beta}_{i_{l(w)}}.
\end{equation}
Now we can use Lemma \ref{lemma1} to prove the following:
\begin{lem}\label{l2}
\[
\bra{1,2,\cdots,k,0,k+1,\cdots,N-1} B_{N-1}(\sigma_1) \cdots B_1(\sigma_{N-1})
= \mathcal{G}^{(\beta)}_{\rho^{-1}_k \omega_N }(\bm{\sigma};\ominus \bm{t}) \bra{\Omega^{(0)}}.
\]
\end{lem}
\begin{proof}
By repeated use of Lemma \ref{lemma1}, one can compute
\begin{equation}\label{l2e1}
\begin{aligned}
&\bra{1,2,\cdots,k,0,k+1,\cdots,N-1} B_{N-1}(\sigma_1) \cdots B_1(\sigma_{N-1})\\
=& \prod_{a=1}^{N-k-1}\prod_{i=1}^{N-a}(\sigma_a \ominus t_i) \prod_{b=N-k}^{N-2} \prod_{j=1}^{N-b-1} (\sigma_b \ominus t_j) \bra{\Omega^{(0)}}.
\end{aligned}
\end{equation}
Notice that $\sigma \ominus t = \frac{\sigma-t}{1+\beta t} = \sigma + \ominus t + \beta \sigma (\ominus t)$, Eq. \eqref{highestGroth} then yields
\[
\prod_{i+j \leq N} (\sigma_i \ominus t_j) = \prod_{i+j \leq N} (\sigma_i + \ominus t_j + \beta \sigma_i (\ominus t_j)) = \mathcal{G}^{(\beta)}_{\omega_N}(\bm{\sigma};\ominus \bm{t}).
\]
Therefore, from \eqref{l2e1} and the definition of the double $\beta$-Grothendieck polynomials Eq. \eqref{didi}-\eqref{otherGroth}, we obtain
\[\begin{aligned}
&\bra{1,2,\cdots,k,0,k+1,\cdots,N-1} B_{N-1}(\sigma_1) \cdots B_1(\sigma_{N-1})\\
= & \frac{\mathcal{G}^{(\beta)}_{\omega_N}(\bm{\sigma};\ominus \bm{t})}{(\sigma_{N-1}\ominus t_1)(\sigma_{N-2}\ominus t_2)\cdots (\sigma_{N-k}\ominus t_k)}\bra{\Omega^{(0)}} \\
=& \partial^\beta_{N-1} \partial^\beta_{N-2} \cdots \partial^\beta_{N-k}\mathcal{G}^{(\beta)}_{\omega_N}(\bm{\sigma};\ominus \bm{t}) \bra{\Omega^{(0)}}\\
=&\mathcal{G}^{(\beta)}_{\omega_N s_{N-k} \cdots s_{N-2} s_{N-1}}(\bm{\sigma};\ominus \bm{t})\bra{\Omega^{(0)}} \\
=&\mathcal{G}^{(\beta)}_{s_k \cdots s_2 s_1 \omega_N}(\bm{\sigma};\ominus \bm{t}) \bra{\Omega^{(0)}} = \mathcal{G}^{(\beta)}_{\rho^{-1}_k \omega_N }(\bm{\sigma};\ominus \bm{t}) \bra{\Omega^{(0)}},
\end{aligned}\]
where in the second to last step, we used $\omega_N s_{N-i} = s_i \omega_N$.
\end{proof}

In order to prove Theorem \ref{thm1}, we need the following lemmas.

\begin{lem}\label{l3pre}
For the states
\begin{equation}\label{psi1}
\bra{\psi} = \bra{n_1,\cdots,n_{a-1},i-1,n_{a+1},\cdots,n_{b-1},i,n_{b+1},\cdots,n_N}
\end{equation}
and
\begin{equation}\label{psi2}
\bra{s_i \psi} = \bra{n_1,\cdots,n_{a-1},i,n_{a+1},\cdots,n_{b-1},i-1,n_{b+1},\cdots,n_N},
\end{equation}
where $n_s < i-1$ for $s \neq a,b$, the following identity holds:
\begin{equation}\label{l3i}
(\bra{s_i\psi} + \beta \bra{\psi}) B_i(x) B_{i-1}(y) = \bra{\psi} B_{i-1}(x) B_i(y).
\end{equation}
\end{lem}
\begin{proof}
We have the following expansions in the natural basis according to \eqref{exdL}
\[
\begin{aligned}
&\bra{\psi}B_{i-1}(x) = \sum_{\bm{m} \in M_{i-1}} C_{\bm{m}}^{(i-1)}\bra{\psi^{(i-1)}_{\bm{m}}},\\ 
&\bra{\psi}B_{i}(x) = \sum_{\bm{m} \in M_{i}} C_{\bm{m}}^{(i)}\bra{\psi^{(i)}_{\bm{m}}},\\ 
&\bra{s_i\psi} B_i(x) = \sum_{\bm{m} \in \bar{M}_i} \bar{C}^{(i)}_{\bm{m}}\bra{\bar{\psi}^{(i)}_{\bm{m}}},
\end{aligned}
\]
where $M_{i-1}, M_i$ and $\bar{M}_i$ are the sets in which $\bm{m} = (m_1,\cdots,m_N)$ takes value such that the corresponding coefficients $C_{\bm{m}}^{(i-1)}$, $C_{\bm{m}}^{(i)}$ and $\bar{C}^{(i)}_{\bm{m}}$ in \eqref{exdL} are nonzero, and $\bra{\psi^{(i-1)}_{\bm{m}}}$, $\bra{\psi^{(i)}_{\bm{m}}}$ and $\bra{\bar{\psi}^{(i)}_{\bm{m}}}$ are members of the natural basis \eqref{basis}. Because $m_a=i-1, m_b=i$ for $\bm{m} \in M_{i-1}$ and $m_a=i, m_b=i-1$ for $\bm{m} \in \bar{M}_{i}$, there is a map $\phi: M_{i-1} \rightarrow \bar{M}_i$ defined by
\[
\phi(\bm{m})_j =  \left\{ \begin{array}{ll}
i, & \mathrm{if}~m_j = i-1,\\
i-1, & \mathrm{if}~m_j = i,\\
m_j, & \mathrm{if}~m_j \neq i~\mathrm{and}~m_j \neq i-1.
\end{array}\right.
\]
$\phi$ is a bijection since $n_s < i-1$ for $s \neq a,b$. Moreover, for $\bm{m} \in M_i$, we have $n_j < m_j$ for $1 \leq j < b$ and $m_b=i$, therefore we have an injective map $\tau: M_i \rightarrow \bar{M}_i$ defined by
\[
\tau(\bm{m})_j =  \left\{ \begin{array}{ll}
i, & 1 \leq j \leq a,\\
i-1, & a+1 \leq j \leq b,\\
m_j, & j>b.
\end{array}\right.
\]
For $\bm{m} \in \mathrm{im} \tau$, i.e. $m_s=i, m_t=i-1$ for $1\leq s \leq a, a+1 \leq t \leq b$, since $n_a = m_a > m_{a+1}$ and $\phi^{-1}(\bm{m})_a < \phi^{-1}(\bm{m})_{a+1}, \tau^{-1}(\bm{m})_a > n_a$, Eq. \eqref{exdL} implies
\begin{equation}\label{CC1}
C^{(i-1)}_{\phi^{-1}(\bm{m})} = (1+\beta(x\ominus t_a)) \bar{C}^{(i)}_{\bm{m}},\quad
C^{(i)}_{\tau^{-1}(\bm{m})} = (x\ominus t_a) \bar{C}^{(i)}_{\bm{m}}.
\end{equation}
On the other hand, if $\bm{m} \not\in \mathrm{im} \tau$, then $m_{a+1}<i-1$, therefore
\begin{equation}\label{CC2}
C^{(i-1)}_{\phi^{-1}(\bm{m})} = \bar{C}^{(i)}_{\bm{m}}.
\end{equation}
Consequently,
\begin{equation}\label{Bact}
\begin{aligned}
&\bra{\psi}B_{i-1}(x) = \sum_{\phi^{-1}(\bm{m}) \in M_{i-1}} C_{\phi^{-1}(\bm{m})}^{(i-1)}\bra{\psi^{(i-1)}_{\phi^{-1}(\bm{m})}}
\\ 
& \quad \quad \quad \quad  \quad= \sum_{\bm{m} \in \mathrm{im} \tau} (1+\beta(x\ominus t_a)) \bar{C}^{(i)}_{\bm{m}} \bra{\psi^{(i-1)}_{\phi^{-1}(\bm{m})}} + \sum_{\bm{m} \not\in \mathrm{im} \tau} \bar{C}^{(i)}_{\bm{m}} \bra{\psi^{(i-1)}_{\phi^{-1}(\bm{m})}},\\
&\bra{\psi}B_i(x)= \sum_{\bar{\bm{m}}\in M_i} C^{(i)}_{\bar{\bm{m}}}\bra{\psi^{(i)}_{\bar{\bm{m}}}} = \sum_{\bar{\bm{m}}\in M_i} (x \ominus t_a) \bar{C}^{(i)}_{\tau(\bar{\bm{m}})}\bra{\psi^{(i)}_{\bar{\bm{m}}}} = \sum_{\bm{m}\in \mathrm{im}\tau} (x \ominus t_a) \bar{C}^{(i)}_{\bm{m}}\bra{\psi^{(i)}_{\tau^{-1}(\bm{m})}},\\
&\bra{s_i \psi}B_{i}(x) = \sum_{\bm{m} \in \bar{M}_{i}} \bar{C}_{\bm{m}}^{(i)}\bra{\bar{\psi}^{(i)}_{\bm{m}}} = \sum_{\bm{m} \in \mathrm{im} \tau} \bar{C}^{(i)}_{\bm{m}} \bra{\bar{\psi}^{(i)}_{\bm{m}}} + \sum_{\bm{m} \not\in \mathrm{im} \tau} \bar{C}^{(i)}_{\bm{m}} \bra{\bar{\psi}^{(i)}_{\bm{m}}},
\end{aligned}
\end{equation}
where, according to \eqref{exdL},
\[
\begin{aligned}
&\bra{\psi^{(i-1)}_{\phi^{-1}(\bm{m})}} = \bra{n_1,\cdots,n_{a-1},i,n_{a+1},\cdots,n_{b-1},n'_{s_1},\cdots,n'_{s_{N-b+1}}},\\
&\bra{\psi^{(i)}_{\tau^{-1}(\bm{m})}} = \bra{n_1,\cdots,n_{a-1},i-1,n_{a+1},\cdots,n_{b-1},n'_{s_1},\cdots,n'_{s_{N-b+1}}},\\
&\bra{\bar{\psi}^{(i)}_{\bm{m}}} = \bra{n_1,\cdots,n_{a-1},i-1,n_{a+1},\cdots,n_{b-1},n'_{s_1},\cdots,n'_{s_{N-b+1}}},
\end{aligned}
\]
with some $n'_{s_1},\cdots,n'_{s_{N-b+1}}$ determined by $\bm{m}$.
From the equations above, we see\\
$\bra{\psi^{(i)}_{\tau^{-1}(\bm{m})}} = \bra{\bar{\psi}^{(i)}_{\bm{m}}}$ and
$\bra{\psi^{(i-1)}_{\phi^{-1}(\bm{m})}}B_i(y) = \bra{\bar{\psi}^{(i)}_{\bm{m}}} B_{i-1}(y)$, from which, together with Eq. \eqref{Bact}, one can compute
\[
\begin{aligned}
&(\bra{s_i\psi}+\beta \bra{\psi})B_i(x) B_{i-1}(y) \\
=&\sum_{\bm{m} \not\in \mathrm{im}\tau}\bar{C}^{(i)}_{\bm{m}}\bra{\bar{\psi}^{(i)}_{\bm{m}}} B_{i-1}(y) + \sum_{\bm{m} \in \mathrm{im}\tau}\bar{C}^{(i)}_{\bm{m}}\bra{\bar{\psi}^{(i)}_{\bm{m}}} B_{i-1}(y) + \beta \sum_{\bm{m} \in \mathrm{im}\tau}(x \ominus t_a) \bar{C}^{(i)}_{\bm{m}}\bra{\bar{\psi}^{(i)}_{\bm{m}}} B_{i-1}(y)\\
=&\sum_{\bm{m} \not\in \mathrm{im}\tau}\bar{C}^{(i)}_{\bm{m}}\bra{\psi^{(i-1)}_{\phi^{-1}(\bm{m})}} B_{i}(y) + \sum_{\bm{m} \in \mathrm{im}\tau}(1+\beta(x \ominus t_a)) \bar{C}^{(i)}_{\bm{m}}\bra{\psi^{(i-1)}_{\phi^{-1}(\bm{m})}} B_{i}(y)\\
=&\sum_{\phi^{-1}(\bm{m})\in M_{i-1}}C^{(i-1)}_{\phi^{-1}(\bm{m})}\bra{\psi^{(i-1)}_{\phi^{-1}(\bm{m})}} B_i(y)\\
=& \bra{\psi}B_{i-1}(x) B_i(y).
\end{aligned}
\]
\end{proof}

Lemma \ref{l3pre} can be generalized to the following:

\begin{lem}\label{l3}
If $\pi \in S_N$ and $l(s_i \pi) > l(\pi)$, then
\[
\begin{aligned}
&(\bra{s_i\pi} + \beta \bra{\pi}) B_{N-1}(\sigma_1) B_{N-2}(\sigma_2) \cdots B_{i+1}(\sigma_{N-i-1})B_i(\sigma_{N-i}) B_{i-1}(\sigma_{N-i+1})\\ =& 
\bra{\pi} B_{N-1}(\sigma_1) B_{N-2}(\sigma_2) \cdots B_{i+1}(\sigma_{N-i-1}) B_{i-1}(\sigma_{N-i}) B_i(\sigma_{N-i+1}).
\end{aligned}
\]
\end{lem}
\begin{proof}
We prove by induction on the number of $B_s$ operators. Assume, for $k \geq i+1$ and $\bra{\varphi} = \bra{n_1,\cdots,n_N}$ with $\{n_1,\cdots,n_N\} = \{ 1,2,\cdots,k-1,0,\cdots,0 \}$, we have
\begin{equation}\label{l3a}
(\bra{s_i\varphi} + \beta \bra{\varphi}) B_{k-1}(\sigma_1) \cdots B_i(\sigma_{k-i}) B_{i-1}(\sigma_{k-i+1}) = \bra{\varphi} B_{k-1}(\sigma_1) \cdots B_{i-1}(\sigma_{k-i}) B_{i}(\sigma_{k-i+1}),
\end{equation}
where $\bra{\varphi}$ and $\bra{s_i \varphi}$ are defined in the same way as Eq. \eqref{psi1} and \eqref{psi2}. The $k=i+1$ case of \eqref{l3a} is valid due to Eq. \eqref{l3i}.
Now let $\bra{\psi}$ be a state of the form
\[
\bra{\psi} = \bra{n_1,\cdots,n_{a-1},i-1,n_{a+1},\cdots,n_{b-1},i,n_{b+1},\cdots,n_N}
\]
with $\{n_1,\cdots,n_N\} = \{ 1,2,\cdots,k,0,\cdots,0 \}$, where $n_a=i-1, n_b=i$.
We expand the following states in the natural basis according to \eqref{exdL}:
\begin{align}
&\bra{\psi} B_k(u) = \sum_{\bm{m}\in M_1}C_{\bm{m}} \bra{\psi_{\bm{m}}} + \sum_{\bm{m}\in M_2} C'_{\bm{m}} \bra{\psi'_{\bm{m}}},\label{Bk}\\
&\bra{s_i \psi} B_k(u) = \sum_{\bm{m} \in \bar{M}_1} \bar{C}_{\bm{m}} \bra{\bar{\psi}_{\bm{m}}}\label{sBk},
\end{align}
where $M_2$ contains $\bm{m}$'s satisfying $i-1 = n_a < m_a = m_{a+1} = \cdots = m_b = i$, and $M_1$ consists of all the other $\bm{m}$'s such that the corresponding coefficients are nonzero. One important difference between these two sets is that $i-1$ remains to the left of $i$ in $\bra{\psi_{\bm{m}}}$, while $i$ is moved to the left of $i-1$ in $\bra{\psi'_{\bm{m}}}$. Since it is not possible to have $i = n_a < m_a = m_{a+1} = \cdots = m_b = i-1$ in the expansion \eqref{sBk}, we have a one-to-one correspondence between $M_1$ and $\bar{M}_1$ given by $\tilde{\phi}:M_1 \rightarrow \bar{M}_1$ with $\tilde{\phi}$ defined by
\[
\tilde{\phi}(\bm{m})_j =  \left\{ \begin{array}{ll}
i, & \mathrm{if}~m_j = i-1,\\
i-1, & \mathrm{if}~m_j = i,\\
m_j, & \mathrm{if}~m_j \neq i~\mathrm{and}~m_j \neq i-1.
\end{array} \right.
\]
Since $\bra{\bar{\psi}_{\tilde{\phi}(\bm{m})}}$ differs from $\bra{\psi_{\bm{m}}}$ by exchange of $i-1$ and $i$, it is easy to see
\begin{equation}\label{psidentity0}
\bra{\bar{\psi}_{\tilde{\phi}(\bm{m})}} = \bra{s_i \psi_{\bm{m}}}.
\end{equation}
Notice that $n_s<i-1$ or $n_s > i$ for $s \neq a,b$, Eq. \eqref{exdL} implies
\begin{equation}\label{Cm1}
C_{\bm{m}} =\left\{ \begin{array}{ll}
(1+\beta (u\ominus t_a)) \bar{C}_{\tilde{\phi}(\bm{m})}, & \mathrm{if}~m_a = i-1, m_b=i, n_s < m_s = m_b ~\mathrm{for}~a<s<b,\\
\bar{C}_{\tilde{\phi}(\bm{m})}, & \mathrm{otherwise}.
\end{array} \right.
\end{equation}
Moreover, there is an injective map $\tilde{\tau}:M_2 \rightarrow M_1$ given by
\[
\tilde{\tau}(\bm{m})_j =  \left\{ \begin{array}{ll}
i-1, & \mathrm{if}~m_j = i~\mathrm{and}~j \leq a,\\
m_j, & \mathrm{otherwise}.
\end{array} \right.
\]
Notice that $\bm{m} \in \mathrm{im}\tilde{\tau}$ is equivalent to $m_a=i-1, m_b=i, n_s<m_s=m_b$ for $a<s<b$, so Eq. \eqref{Cm1} can be rewritten as
\begin{equation}\label{CmCbar}
C_{\bm{m}} =\left\{ \begin{array}{ll}
(1+\beta (u\ominus t_a)) \bar{C}_{\tilde{\phi}(\bm{m})}, & \mathrm{if}~\bm{m} \in \mathrm{im} \tilde{\tau},\\
\bar{C}_{\tilde{\phi}(\bm{m})}, & \mathrm{otherwise}.
\end{array} \right.
\end{equation}
In addition, if $\bm{m} \in \mathrm{im} \tilde{\tau}$, then
\begin{equation}\label{psidentity1}
\bra{\bar{\psi}_{\tilde{\phi}(\bm{m})}} = \bra{\psi'_{\tilde{\tau}^{-1}(\bm{m})}},
\end{equation}
because on both sides, $i-1$ is moved to the $a$th site, $i$ is moved to the $c$th site, where $c = \mathrm{max}\{ l<a ~|~ n_l=m_l \}$, and the other $n_s$'s are moved in the same way, and
\begin{equation}\label{psidentity2}
C'_{\tilde{\tau}^{-1}(\bm{m})} = (u\ominus t_a) \bar{C}_{\tilde{\phi}(\bm{m})},
\end{equation}
because $\tilde{\phi}(\bm{m})_a = n_a = i$, $\tilde{\tau}^{-1}(\bm{m})_a > n_a = i-1$, and $n_s < i-1$ or $n_s > i$ for $s \neq a,b$.
Therefore
\begin{equation}\label{CmComp}
\begin{aligned}
&\sum_{\bm{m}\in M_1} \bar{C}_{\tilde{\phi}(\bm{m})}\bra{\bar{\psi}_{\tilde{\phi}(\bm{m})}} + \beta \sum_{\bm{m} \in M_2} C'_{\bm{m}} \bra{\psi'_{\bm{m}}}\\
=& \sum_{\bm{m} \in M_1\backslash \mathrm{im}\tilde{\tau}} \bar{C}_{\tilde{\phi}(\bm{m})}\bra{\bar{\psi}_{\tilde{\phi}(\bm{m})}} + \sum_{\bm{m} \in \mathrm{im}\tilde{\tau}} \bar{C}_{\tilde{\phi}(\bm{m})} \bra{\bar{\psi}_{\tilde{\phi}(\bm{m})}} + \beta \sum_{\bm{m} \in M_2} C'_{\bm{m}}\bra{\psi'_{\bm{m}}}\\
=& \sum_{\bm{m} \in M_1\backslash \mathrm{im}\tilde{\tau}} \bar{C}_{\tilde{\phi}(\bm{m})}\bra{s_i \psi_{\bm{m}}} + \sum_{\bm{m} \in \mathrm{im}\tilde{\tau}} \bar{C}_{\tilde{\phi}(\bm{m})} \bra{s_i \psi_{\bm{m}}} + \beta (u \ominus t_a) \sum_{\bm{m} \in \mathrm{im}\tilde{\tau}} \bar{C}_{\tilde{\phi}(\bm{m})}\bra{s_i \psi_{\bm{m}}}\\
=& \sum_{\bm{m} \in M_1} C_{\bm{m}} \bra{s_i \psi_{\bm{m}}},
\end{aligned}
\end{equation}
where the second equality is due to Eq. \eqref{psidentity0}, \eqref{psidentity1} and \eqref{psidentity2}, and in the last equality we used Eq. \eqref{CmCbar}.
Consequently,
\begin{equation}\label{Bkexpansion}
\begin{aligned}
&(\bra{s_i \psi} + \beta \bra{\psi}) B_k(u)\\
=& \sum_{\bm{m} \in M_1} \bar{C}_{\tilde{\phi}(\bm{m})}\bra{\bar{\psi}_{\tilde{\phi}(\bm{m})}} + \beta \sum_{\bm{m} \in M_2} C'_{\bm{m}} \bra{\psi'_{\bm{m}}} + \beta \sum_{\bm{m} \in M_1} C_{\bm{m}}\bra{\psi_{\bm{m}}}\\
=& \sum_{\bm{m} \in M_1} C_{\bm{m}} (\bra{s_i \psi_{\bm{m}}}+ \beta \bra{\psi_{\bm{m}}}),
\end{aligned}
\end{equation}
where the second equality is due to \eqref{CmComp}.
From the inductive hypothesis Eq. \eqref{l3a}, we have
\begin{equation}\label{indhypcon}
(\bra{s_i \psi_{\bm{m}}}+ \beta \bra{\psi_{\bm{m}}}) B_{k-1}(\sigma_1) \cdots B_{i}(\sigma_{k-i}) B_{i-1}(\sigma_{k-i+1})
=\bra{\psi_{\bm{m}}} B_{k-1}(\sigma_1) \cdots B_{i-1}(\sigma_{k-i}) B_i(\sigma_{k-i+1}),
\end{equation}
which results in
\begin{equation}\label{l3c}
\begin{aligned}
&(\bra{s_i \psi}+ \beta \bra{\psi}) B_k(u) B_{k-1}(\sigma_1) \cdots B_{i}(\sigma_{k-i}) B_{i-1}(\sigma_{k-i+1})\\
=&\sum_{\bm{m} \in M_1} C_{\bm{m}} (\bra{s_i \psi_{\bm{m}}}+ \beta \bra{\psi_{\bm{m}}}) B_{k-1}(\sigma_1) \cdots B_{i}(\sigma_{k-i}) B_{i-1}(\sigma_{k-i+1})\\
=&\sum_{\bm{m} \in M_1} C_{\bm{m}} \bra{\psi_{\bm{m}}} B_{k-1}(\sigma_1) \cdots B_{i-1}(\sigma_{k-i}) B_i(\sigma_{k-i+1})\\
=&\bra{\psi} B_{k}(u) B_{k-1}(\sigma_1) \cdots B_{i-1}(\sigma_{k-i}) B_i(\sigma_{k-i+1}),
\end{aligned}
\end{equation}
where the first equality is due to \eqref{Bkexpansion}, the second equality is due to \eqref{indhypcon}, and in the last step we used Eq. \eqref{Bk} and the fact that, for $\bm{m} \in M_2$, $i$ is to the left of $i-1$ in $\bra{\psi'_{\bm{m}}}$, so (notice that $B_l$ with $l>i$ cannot move $i-1$ to the left of $i$ because the only way to do this is to have $i<m_s = m_t = i-1$ for some $s<t$, which is impossible)
\[
\bra{\psi'_{\bm{m}}} B_{k-1}(\sigma_1) \cdots B_{i-1}(\sigma_{k-i}) B_i(\sigma_{k-i+1}) = 0,\quad \bm{m} \in M_2,
\]
because it is impossible to have $i < m_s = m_{s+1} = \cdots = m_t = i-1$ when $B_{i-1}(\sigma_{k-i})$ acts on the state.

Combining Eq. \eqref{l3i}, \eqref{l3a} and \eqref{l3c}, the lemma is proved by induction.
\end{proof}

In order to relate the quantum states of the integrable system to the double $\beta$-Grothendieck polynomials, we need the commutation relations among the different $B_i$ operators. From Eq. \eqref{Rijkl} and \eqref{comBB}, we obtain the commutation relations
\begin{align}
&B_i(x) B_i(y) = B_i(y) B_i(x),\label{comme} \\
&B_i(x) B_j(y) = \frac{B_j(y)B_i(x)-B_j(x)B_i(y)}{y\ominus x} = (1+\beta x)\frac{B_j(x)B_i(y)-B_j(y)B_i(x)}{x-y},\quad i < j \label{comm}.
\end{align}
Let us define
\begin{equation}\label{dbar}
\bar{\partial}^\beta_i = (1+\beta x_i) \partial_i,
\end{equation}
where $\partial_i$ is the ordinary divided difference operator defined by $\partial^{\beta}_i|_{\beta=0}$.
It is easy to check that
\begin{equation}\label{dbar+}
\bar{\partial}^\beta_i = \partial^\beta_i + \beta.
\end{equation}
Similar to Eq. \eqref{ds}, for $w = s_{i_1} s_{i_2} \cdots s_{i_{l(w)}}$, we define
\begin{equation}\label{dbars}
\bar{\partial}^{\beta}_w := \bar{\partial}^{\beta}_{i_1} \bar{\partial}^{\beta}_{i_2} \cdots \bar{\partial}^{\beta}_{i_{l(w)}}.
\end{equation}

From Eq. \eqref{comm}, \eqref{dbar} and Lemma \ref{l3}, we can prove
\begin{lem}\label{l4}
For $\pi \in S_N$, if $l(s_i \pi) > l(\pi)$, then
\[
\bra{s_i \pi} B_{N-1}(\sigma_1) \cdots B_i(\sigma_{N-i}) B_{i-1}(\sigma_{N-i+1}) = 
\partial^\beta_{N-i} \bra{\pi} B_{N-1}(\sigma_1) \cdots B_i(\sigma_{N-i}) B_{i-1}(\sigma_{N-i+1}).
\]
\end{lem}
\begin{proof}
\[
\begin{aligned}
&(\bra{s_i \pi} + \beta\bra{\pi}) B_{N-1}(\sigma_1) \cdots B_i(\sigma_{N-i}) B_{i-1}(\sigma_{N-i+1})\\
=& \bra{\pi} B_{N-1}(\sigma_1) \cdots B_{i-1}(\sigma_{N-i}) B_i(\sigma_{N-i+1})\\
=& (\partial^\beta_{N-i} + \beta) \bra{\pi} B_{N-1}(\sigma_1) \cdots B_i(\sigma_{N-i}) B_{i-1}(\sigma_{N-i+1}),
\end{aligned}
\]
where the first equality is due to Lemma \ref{l3}, and the second equality is due to Eq. \eqref{comm}, \eqref{dbar} and \eqref{dbar+}.
\end{proof}

\begin{thm}\label{thm1}
For $\pi \in S_N$,
\[
B_{N-1}(\sigma_1)B_{N-2}(\sigma_2)\cdots B_1(\sigma_{N-1}) \ket{\Omega^{(0)}} = \sum_{w \in S_N} \mathcal{G}^{(\beta)}_w (\bm{\sigma};\ominus \bm{t}) \ket{\omega_N w^{-1}},
\]
where $\omega_N$ is the permutation in $S_N$ of maximal length.
\end{thm}
\begin{proof}
Lemma \ref{l0} suggests that we only need to show
\begin{equation}\label{th1}
\bra{\omega_N w^{-1}} B_{N-1}(\sigma_1) \cdots B_1(\sigma_{N-1}) \ket{\Omega^{(0)}} = \mathcal{G}^{(\beta)}_w (\bm{\sigma};\ominus \bm{t}).
\end{equation}
Lemma \ref{l2} shows Eq. \eqref{th1} is valid for $\omega_N w^{-1} = \rho_k, k=0,1,\cdots,N-1$. Because any permutation $w \in S_N$ can be written as $w=\pi \rho_k$ for some $\pi \in S_{N-1}$ acting on the last $N-1$ elements and some $k=0,\cdots,N-1$, it remains to show the validity of Eq. \eqref{th1} for $\omega_N w^{-1} = \pi \rho_k$ with $\pi \in S_{N-1}$ acting on $\{ 1,2,\cdots,N-1 \}$. We show this by induction on the length of $\pi$. Let us assume
\[
\bra{\tilde{\pi} \rho_k} B_{N-1}(\sigma_1) \cdots B_2(\sigma_{N-2}) B_1(\sigma_{N-1}) \ket{\Omega^{(0)}} = \mathcal{G}^{(\beta)}_{ (\tilde{\pi} \rho_k)^{-1} \omega_N}(\bm{\sigma};\ominus \bm{t})
\]
for $\tilde{\pi} \in S_{N-1}$ with $l(\tilde{\pi}) \leq m-1$ for some $m \geq 1$. Now, for $\pi \in S_{N-1}$ with $l(\pi) = m$, there exists $\tilde{\pi} \in S_{N-1}$ with $l(\tilde{\pi}) = m-1$, such that $\pi = s_i \tilde{\pi}$. Then
\[
\begin{aligned}
&\bra{\pi \rho_k} B_{N-1}(\sigma_1) \cdots B_1(\sigma_{N-1}) \ket{\Omega^{(0)}}\\
=&\bra{s_i \tilde{\pi} \rho_k} B_{N-1}(\sigma_1) \cdots B_1(\sigma_{N-1}) \ket{\Omega^{(0)}}
\\
=&\partial^\beta_{N-i} \bra{\tilde{\pi} \rho_k} B_{N-1}(\sigma_1) \cdots B_1(\sigma_{N-1}) \ket{\Omega^{(0)}}\\
=&\partial^\beta_{N-i} \mathcal{G}^{(\beta)}_{\rho_k^{-1} \tilde{\pi}^{-1} \omega_N} (\bm{\sigma};\ominus \bm{t})\\
=&\mathcal{G}^{(\beta)}_{\rho_k^{-1} \tilde{\pi}^{-1} \omega_N s_{N-i}} (\bm{\sigma};\ominus \bm{t}) = \mathcal{G}^{(\beta)}_{\rho_k^{-1} \tilde{\pi}^{-1} s_i \omega_N} (\bm{\sigma};\ominus \bm{t})\\
=&\mathcal{G}^{(\beta)}_{(\pi \rho_k)^{-1} \omega_N} (\bm{\sigma};\ominus \bm{t}),
\end{aligned}
\]
where the second equality is due to Lemma \ref{l4}.
\end{proof}

\begin{remark}
Theorem \ref{thm1} can also be proved following \cite{BFHTW}, but here we have adopted a different approach which does not make use of the 2d lattice models.
\end{remark}

\begin{lem}\label{l5}
For any $w \in S_{N-1}$,
\[
B_{w(1)}(\sigma_1) \cdots B_{w(N-1)}(\sigma_{N-1}) = \bar{\partial}^\beta_{w^{-1}\omega_{N-1}} B_{N-1}(\sigma_1) \cdots B_1(\sigma_{N-1}).
\] 
\end{lem}
\begin{proof}
If we write $w=\pi \omega_{N-1}$, then
\begin{equation}\label{BBBB}
B_{w(1)}(\sigma_1) \cdots B_{w(N-1)}(\sigma_{N-1}) = B_{\pi(N-1)}(\sigma_1) \cdots B_{\pi(1)}(\sigma_{N-1}).
\end{equation}
Let $\pi = s_{i_1} s_{i_2} \cdots s_{i_p}$, where $p = l(\pi)$. Then we can apply the transpositions $s_{N-i_p-1}, \cdots, s_{N-i_2-1}$, $s_{N-i_1-1}$ successively to the subscripts of the right hand side of \eqref{BBBB} to bring them to the strictly descending order, where $s_j$ exchanges the $j$th and $(j+1)$th subindices at each step. Since we need at least $p$ transpositions to complete the ordering, every transposition must bring a smaller subindex to the right of a neighboring larger subindex. According to \eqref{comm}, the application of $s_j$ must be accompanied by the action of $\bar{\partial}^\beta_j$ to the right hand side of \eqref{BBBB} in order to keep the equality. Therefore,
\[
B_{\pi(N-1)}(\sigma_1) \cdots B_{\pi(1)}(\sigma_{N-1}) = \bar{\partial}^\beta_{N-i_p-1} \cdots \bar{\partial}^\beta_{N-i_2-1} \bar{\partial}^\beta_{N-i_1-1} B_{N-1}(\sigma_1) \cdots B_1(\sigma_{N-1}).
\]
Because of the identity
\[
s_{N-i_p-1} \cdots s_{N-i_2-1} s_{N-i_1-1} = \omega_{N-1} s_{i_p} \cdots s_{i_2} s_{i_1} \omega_{N-1} = \omega_{N-1} \pi^{-1} \omega_{N-1} = w^{-1} \omega_{N-1},
\]
the proof is complete.
\end{proof}
\begin{remark}
Lemma \ref{l5} allows us to deduce the matrix elements of $B_{w(1)}(\sigma_1) \cdots B_{w(N-1)}(\sigma_{N-1})$ from Theorem \ref{thm1}.
\end{remark}

\begin{example}\label{ex1}
For $N=3$, we have, by Eq. \eqref{exdR} or Theorem \ref{thm1}, 
\[
\begin{aligned}
&B_2(\sigma_1) B_1(\sigma_2) \ket{000}\\
=& \ket{210} + (\sigma_1 \ominus t_1) \ket{120} + [(\sigma_1 \ominus t_1) + (\sigma_2 \ominus t_2) + \beta (\sigma_1 \ominus t_1)(\sigma_2 \ominus t_2)] \ket{201}\\
&+(\sigma_1 \ominus t_1)(\sigma_1 \ominus t_2)\ket{102} + (\sigma_1 \ominus t_1)(\sigma_2 \ominus t_1)\ket{021} + (\sigma_1 \ominus t_1)(\sigma_1 \ominus t_2)(\sigma_2 \ominus t_1)\ket{012}\\
=&\mathcal{G}^{(\beta)}_{123}(\sigma_1,\sigma_2;\ominus t_1,\ominus t_2) \ket{210}+\mathcal{G}^{(\beta)}_{213}(\sigma_1,\sigma_2;\ominus t_1,\ominus t_2) \ket{120} + 
\mathcal{G}^{(\beta)}_{132}(\sigma_1,\sigma_2;\ominus t_1,\ominus t_2) \ket{201}\\
&+\mathcal{G}^{(\beta)}_{312}(\sigma_1,\sigma_2;\ominus t_1,\ominus t_2) \ket{102} +
\mathcal{G}^{(\beta)}_{231}(\sigma_1,\sigma_2;\ominus t_1,\ominus t_2) \ket{021} +
\mathcal{G}^{(\beta)}_{321}(\sigma_1,\sigma_2;\ominus t_1,\ominus t_2) \ket{012},
\end{aligned}
\]
while
\begin{equation}\label{exam}
\begin{aligned}
&B_1(\sigma_1)B_2(\sigma_2) \ket{000} \\
=& (1+\beta(\sigma_1 \ominus t_1)) \ket{120} + [(\sigma_2 \ominus t_1) (1+\beta(\sigma_1 \ominus t_2))+(\sigma_1 \ominus t_2)(1+\beta(\sigma_1\ominus t_1))] \ket{102}\\
&+ (\sigma_1 \ominus t_1)(\sigma_2 \ominus t_1)(1+\beta(\sigma_1 \ominus t_2)) \ket{012}.
\end{aligned}
\end{equation}
It is easy to verify that Eq. \eqref{exam} is consistent with Lemma \ref{l5} since we have
\[
\begin{aligned}
&\bar{\partial}^\beta_1 \mathcal{G}^{(\beta)}_{213}(\sigma_1,\sigma_2;\ominus t_1,\ominus t_2) = 1+\beta (\sigma_1 \ominus t_1),\\ 
&\bar{\partial}^\beta_1 \mathcal{G}^{(\beta)}_{321}(\sigma_1,\sigma_2;\ominus t_1,\ominus t_2) = (\sigma_1 \ominus t_1)(\sigma_2 \ominus t_1)(1+\beta (\sigma_1 \ominus t_2)),\\
&\bar{\partial}^\beta_1 \mathcal{G}^{(\beta)}_{312}(\sigma_1,\sigma_2;\ominus t_1,\ominus t_2) =(\sigma_2 \ominus t_1) (1+\beta(\sigma_1 \ominus t_2))+(\sigma_1 \ominus t_2)(1+\beta(\sigma_1\ominus t_1)),\\
&\bar{\partial}^\beta_1 \mathcal{G}^{(\beta)}_{231}(\sigma_1,\sigma_2;\ominus t_1,\ominus t_2)=\bar{\partial}^\beta_1 \mathcal{G}^{(\beta)}_{132}(\sigma_1,\sigma_2;\ominus t_1,\ominus t_2)=\bar{\partial}^\beta_1 \mathcal{G}^{(\beta)}_{123}(\sigma_1,\sigma_2;\ominus t_1,\ominus t_2)=0.
\end{aligned}
\]
\end{example}

Now we would like to apply Theorem \ref{thm1} and Lemma \ref{l5} to the Bethe Ansatz state \eqref{BEstate}. First, we show the following 
\begin{thm}\label{cauchy}
For any $w \in S_n$, we have the following form of the Cauchy identity
\[
\mathcal{G}^{(\beta)}_{w}(\bm{x};\ominus\bm{t}) = \sum_{v \in S_n} \mathcal{G}^{(\beta)}_{v}(\bm{x};\ominus\bm{z}) \bar{\partial}^\beta_v \mathcal{G}^{(\beta)}_{w}(\bm{z};\ominus\bm{t}).
\]
\end{thm}
\begin{proof}
The generalized Cauchy identity Eq. \eqref{Gcauchy} can be rewritten as (notice that $\ominus \ominus z = z$)
\begin{equation}\label{id1}
\mathcal{G}^{(\beta)}_{w}(\ominus\bm{t};\bm{x}) = \sum_{v \in S_n} \mathcal{G}^{(\beta)}_{v}(\bm{x};\ominus\bm{z}) \mathcal{G}^{(\beta)}_{v^{-1},w}(\ominus\bm{t};\bm{z}).
\end{equation}
From Eq. \eqref{Gswitch}, \eqref{dG1} and \eqref{dG2}, we get
\[
\mathcal{G}^{(\beta)}_{v^{-1},w}(\ominus\bm{t};\bm{z}) = \bar{\partial}^\beta_{v,\bm{z}} \mathcal{G}^{(\beta)}_{1,w} (\ominus\bm{t};\bm{z}) = \bar{\partial}^\beta_{v} \mathcal{G}^{(\beta)}_{w^{-1}} (\bm{z};\ominus\bm{t}).
\]
Plugging the identity above in \eqref{id1}, we arrive at
\[
\mathcal{G}^{(\beta)}_{w^{-1}}(\bm{x};\ominus\bm{t}) = \sum_{v \in S_n} \mathcal{G}^{(\beta)}_{v}(\bm{x};\ominus\bm{z}) \bar{\partial}^\beta_{v}\mathcal{G}^{(\beta)}_{w^{-1}}(\bm{z};\ominus\bm{t}).
\]
Then the theorem is proved by replacing $w^{-1}$ with $w$.
\end{proof}

Let $L_n(\bm{t},\beta)$ be the submodule of $\mathbb{C}(\bm{t},\beta)[x_1,\cdots,x_n]$ defined by
\begin{equation}\label{Ln}
L_n(\bm{t},\beta) = \mathrm{span}_{\mathbb{C}(\bm{t},\beta)}\{ x_1^{i_1}x_2^{i_2} \cdots x_n^{i_n}~|~0 \leq i_j \leq n-j \}.
\end{equation}
\begin{cor}\label{interpolation}
For any $F(\bm{x};\bm{t},\beta) \in L_n(\bm{t},\beta)$, we have
\begin{equation}\label{Finter}
F(\bm{x};\bm{t},\beta) = \sum_{v \in S_n} \mathcal{G}^{(\beta)}_{v}(\bm{x};\ominus\bm{z}) \bar{\partial}^\beta_{v}F(\bm{z};\bm{t},\beta).
\end{equation}
\end{cor}
\begin{proof}
Since $\{ \mathcal{G}^{(\beta)}_w(\bm{x};\ominus \bm{t})~|~w \in S_n \}$ is a $\mathbb{C}(\bm{t},\beta)$-basis of $L_n(\bm{t},\beta)$ (as are the Schubert and Grothendieck polynomials \cite[Proposition 2.7]{LT06}), Eq. \eqref{Finter} follows from Theorem \ref{cauchy} and linearity of $\bar{\partial}^\beta_v$.
\end{proof}

In section \ref{sec:QK}, we have seen that $\sigma^{(N-i)}_a$ ($\beta = 0$) or $1-\sigma^{(N-i)}_a$ ($\beta = -1$), $a=1,\cdots, i$, can be identified with the Chern roots of $\mathcal{S}_i$ when $\sigma^{(N-i)}_a$'s satisfy the Bethe Ansatz equations Eq. \eqref{BAE1}-\eqref{BAE3}. We will show that the Bethe Ansatz state \eqref{BEstate} generates the double $\beta$-Grothendieck polynomials when this identification is implemented. For this purpose, we need the following
\begin{lem}\label{l7}
For $\pi \in S_n$, if $e_l(\sigma_1,\cdots,\sigma_{n-1})$ is identified with $e_l(x_1,\cdots,x_{n-1})$ for every $l=1,\cdots,n-1$, then
\[
\begin{aligned}
&\sum_{w\in S_{n-1}} \mathcal{G}^{(\beta)}_w(x_1,\cdots,x_{n-2};\ominus\sigma_1,\cdots,\ominus\sigma_{n-2}) \cdot \bar{\partial}^\beta_w \mathcal{G}^{(\beta)}_{\pi}(\sigma_1,\cdots,\sigma_{n-1};\ominus t_1,\cdots,\ominus t_{n-1})\\
=&\mathcal{G}^{(\beta)}_\pi (x_1,\cdots,x_{n-1};\ominus t_1,\cdots,\ominus t_{n-1}).
\end{aligned}
\]
\end{lem}
\begin{proof}
Let $L_n(\bm{t},\beta)$ be the submodule defined by \eqref{Ln}. Then, since $\{ \mathcal{G}^{(\beta)}_w(\bm{x};\ominus \bm{t})~|~w \in S_n \}$ and \\$\{ e_{i_1}(x_1)e_{i_2}(x_1,x_2)\cdots e_{i_{n-1}}(x_1,\cdots,x_{n-1})~|~0\leq i_j \leq j \}$ are both $\mathbb{C}(\bm{t},\beta)$-bases of $L_n(\bm{t},\beta)$, we have the expansion
\[
\begin{aligned}
&\mathcal{G}^{(\beta)}_\pi (\sigma_1,\cdots,\sigma_{n-1};\ominus t_1,\cdots,\ominus t_{n-1})
\\
=& \sum_{i_1,i_2,\cdots,i_{n-1}} C^{\pi}_{i_1 \cdots i_{n-1}}(\bm{t},\beta) e_{i_1}(\sigma_1)e_{i_2}(\sigma_1,\sigma_2)\cdots e_{i_{n-1}}(\sigma_1,\cdots,\sigma_{n-1}),
\end{aligned}
\]
where $C^{\pi}_{i_1 \cdots i_{n-1}}(\bm{t},\beta) \in \mathbb{C}(\bm{t},\beta)$. 

Therefore
\[
\begin{aligned}
&\sum_{w\in S_{n-1}} \mathcal{G}^{(\beta)}_w(x_1,\cdots,x_{n-2};\ominus\sigma_1,\cdots,\ominus\sigma_{n-2}) \cdot \bar{\partial}^\beta_w \mathcal{G}^{(\beta)}_{\pi}(\sigma_1,\cdots,\sigma_{n-1};\ominus t_1,\cdots,\ominus t_{n-1})\\
=&\sum_{w\in S_{n-1}} \mathcal{G}^{(\beta)}_w(x_1,\cdots,x_{n-2};\ominus\sigma_1,\cdots,\ominus\sigma_{n-2}) \\
&\quad \quad \quad \cdot \bar{\partial}^\beta_w \left[ 
\sum_{i_1,i_2,\cdots,i_{n-1}} C^{\pi}_{i_1 \cdots i_{n-1}}(\bm{t},\beta) e_{i_1}(\sigma_1)e_{i_2}(\sigma_1,\sigma_2)\cdots e_{i_{n-1}}(\sigma_1,\cdots,\sigma_{n-1})\right]\\
=&\sum_{i_1,i_2,\cdots,i_{n-1}} C^{\pi}_{i_1 \cdots i_{n-1}}(\bm{t},\beta)
e_{i_{n-1}}(\sigma_1,\cdots,\sigma_{n-1})\\
&\cdot \left[
\sum_{w\in S_{n-1}} \mathcal{G}^{(\beta)}_w(x_1,\cdots,x_{n-2};\ominus\sigma_1,\cdots,\ominus\sigma_{n-2}) \cdot \bar{\partial}^{\beta}_w \left(e_{i_1}(\sigma_1)e_{i_2}(\sigma_1,\sigma_2)\cdots e_{i_{n-2}}(\sigma_1,\cdots,\sigma_{n-2})\right) \right] \\
=&\sum_{i_1,i_2,\cdots,i_{n-1}} C^\pi_{i_1 \cdots i_{n-1}}(\bm{t},\beta) e_{i_1}(x_1)e_{i_2}(x_1,x_2)\cdots e_{i_{n-2}}(x_1,\cdots,x_{n-2}) e_{i_{n-1}}(\sigma_1,\cdots,\sigma_{n-1})\\
=&\sum_{i_1,i_2,\cdots,i_{n-1}} C^\pi_{i_1 \cdots i_{n-1}}(\bm{t},\beta) e_{i_1}(x_1)e_{i_2}(x_1,x_2)\cdots e_{i_{n-2}}(x_1,\cdots,x_{n-2}) e_{i_{n-1}}(x_1,\cdots,x_{n-1})\\
=&\mathcal{G}^{(\beta)}_\pi (x_1,\cdots,x_{n-1};\ominus t_1,\cdots,\ominus t_{n-1}),
\end{aligned}
\]
where the second equality is due to $\bar{\partial}^\beta_i(f\cdot g) = f (\bar{\partial}^\beta_i g)$ if $f$ is symmetric in $(\sigma_1,\cdots,\sigma_{n-1})$, the third equality is due to Corollary \ref{interpolation}, and $e_{i_{n-1}}(\sigma_1,\cdots,\sigma_{n-1})$ is identified with $e_{i_{n-1}}(x_1,\cdots,x_{n-1})$ in the fourth equality.
\end{proof}

\begin{example}
Notice that $\mathcal{G}^{(\beta)}_{12}(x_1;\ominus\sigma_1)=1, \mathcal{G}^{(\beta)}_{21}(x_1;\ominus\sigma_1)=x_1 \ominus \sigma_1$. In the case $n=3$, following Example \ref{ex1}, one can check
\[
\begin{aligned}
&\mathcal{G}^{(\beta)}_{12}(x_1;\ominus\sigma_1) \mathcal{G}^{(\beta)}_{321}(\sigma_1,\sigma_2;\ominus t_1,\ominus t_2) +
\mathcal{G}^{(\beta)}_{21}(x_1;\ominus\sigma_1) \bar{\partial}^{\beta}_1 \mathcal{G}^{(\beta)}_{321}(\sigma_1,\sigma_2;\ominus t_1,\ominus t_2) \\
=&(\sigma_1 \ominus t_1) (\sigma_2 \ominus t_1)(x_1 \ominus t_2)
=(x_1 \ominus t_1) (x_2 \ominus t_1)(x_1 \ominus t_2)\\
=&\mathcal{G}^{(\beta)}_{321}(x_1,x_2;\ominus t_1,\ominus t_2),
\end{aligned}
\]
\[
\begin{aligned}
&\mathcal{G}^{(\beta)}_{12}(x_1;\ominus\sigma_1) \mathcal{G}^{(\beta)}_{231}(\sigma_1,\sigma_2;\ominus t_1,\ominus t_2) = (\sigma_1 \ominus t_1) (\sigma_2 \ominus t_1) = (x_1 \ominus t_1) (x_2 \ominus t_1)\\
=&\mathcal{G}^{(\beta)}_{231}(x_1,x_2;\ominus t_1,\ominus t_2),
\end{aligned}
\]
\[
\begin{aligned}
&\mathcal{G}^{(\beta)}_{12}(x_1;\ominus\sigma_1) \mathcal{G}^{(\beta)}_{312}(\sigma_1,\sigma_2;\ominus t_1,\ominus t_2) +
\mathcal{G}^{(\beta)}_{21}(x_1;\ominus\sigma_1) \bar{\partial}^{\beta}_1 \mathcal{G}^{(\beta)}_{312}(\sigma_1,\sigma_2;\ominus t_1,\ominus t_2) \\
=&\frac{x_1(\sigma_1+\sigma_2-t_1-t_2)+(t_1 t_2-\sigma_1 \sigma_2)}{(1+\beta t_1)(1+\beta t_2)}
=(x_1 \ominus t_1) (x_2 \ominus t_2)\\
=&\mathcal{G}^{(\beta)}_{312}(x_1,x_2;\ominus t_1,\ominus t_2),
\end{aligned}
\]
\[
\begin{aligned}
&\mathcal{G}^{(\beta)}_{12}(x_1;\ominus\sigma_1) \mathcal{G}^{(\beta)}_{132}(\sigma_1,\sigma_2;\ominus t_1,\ominus t_2) =
\frac{x_1+x_2-t_1-t_2+\beta (x_1 x_2-t_1 t_2)}{(1+\beta t_1)(1+\beta t_2)}\\
=&\mathcal{G}^{(\beta)}_{132}(x_1,x_2;\ominus t_1,\ominus t_2),
\end{aligned}
\]
\[
\begin{aligned}
&\mathcal{G}^{(\beta)}_{12}(x_1;\ominus\sigma_1) \mathcal{G}^{(\beta)}_{213}(\sigma_1,\sigma_2;\ominus t_1,\ominus t_2) +
\mathcal{G}^{(\beta)}_{21}(x_1;\ominus\sigma_1) \bar{\partial}^{\beta}_1 \mathcal{G}^{(\beta)}_{213}(\sigma_1,\sigma_2;\ominus t_1,\ominus t_2) \\
=&(\sigma_1 \ominus t_1)+(x_1 \ominus \sigma_1)(1+\beta(\sigma_1 \ominus t_1))
=x_1 \ominus t_1\\
=&\mathcal{G}^{(\beta)}_{213}(x_1,x_2;\ominus t_1,\ominus t_2),
\end{aligned}
\]
where the following identifications have been made
\[
e_1(\sigma_1,\sigma_2)=e_1(x_1,x_2),\quad e_2(\sigma_1,\sigma_2)=e_2(x_1,x_2).
\]
\end{example}

Finally, we show that the Bethe state satisfies the following
\begin{thm}\label{thm2}
When $e_l(\sigma^{(i)}_1,\cdots,\sigma^{(i)}_{N-i})$ is identified with $e_l(x_1,\cdots,x_{N-i})$ for all $i=1,\cdots,N-1$ and $l=1,\cdots,N-i$, the Bethe Ansatz state \eqref{BEstate} in the case $k_i = N-i$ has the expansion
\begin{equation}\label{thm2exp}
\ket{\psi^{(0)}} = \sum_{w \in S_N} \mathcal{G}^{(\beta)}_w (x_1,\cdots,x_{N-1};\ominus t_1,\cdots,\ominus t_{N-1}) \ket{\omega_N w^{-1}}.
\end{equation}
\end{thm}
\begin{proof}
We prove by induction on the number of sites. For a two-site system, it is easy to check that, when $\sigma^{(1)}_1 = x_1$,
\[
\begin{aligned}
&\bra{01} B_1(\sigma^{(1)}_1) \ket{00} = \sigma^{(1)}_1 \ominus t_1 = x_1 \ominus t_1 = \mathcal{G}^{(\beta)}_{21}(x_1;\ominus t_1), \\
&\bra{10} B_1(\sigma^{(1)}_1) \ket{00} = 1 = \mathcal{G}^{(\beta)}_{12}(x_1;\ominus t_1),
\end{aligned}
\]
so Eq. \eqref{thm2exp} holds for $N=2$. Now assume Eq. \eqref{thm2exp} holds for $N=2,3,\cdots,m$ with $m \geq 2$, then we have, from Eq. \eqref{psis} and \eqref{psim}, for $w_l \in S_l$ (notice that $\ket{\psi^{(N-s)}}$ is the Bethe state of a GL$(s)$ system with $s$ sites),
\[
\begin{aligned}
\braket{\omega_2 w_2^{-1}|\psi^{(N-2)}} &= \mathcal{G}^{(\beta)}_{w_2}(x_1;\ominus \sigma_1^{(N-2)}),\\
\braket{\omega_3 w_3^{-1}|\psi^{(N-3)}} &= \mathcal{G}^{(\beta)}_{w_3}(x_1,x_2;\ominus \sigma_1^{(N-3)},\ominus\sigma_2^{(N-3)}),\\
&~~\vdots \\
\braket{\omega_m w_m^{-1}|\psi^{(N-m)}} &= \mathcal{G}^{(\beta)}_{w_m}(x_1,\cdots,x_{m-1};\ominus \sigma_1^{(N-m)},\cdots,\ominus\sigma_{m-1}^{(N-m)}).
\end{aligned}
\]
Eq. \eqref{psis} and \eqref{psim} yield
\begin{equation}\label{psiplus}
\ket{\psi^{(N-m-1)}} = \sum_{v_m \in S_m} \braket{v_m | \psi^{(N-m)}} B^{(N-m-1)}_{v_m(1)}(\sigma^{(N-m)}_1) \cdots B^{(N-m-1)}_{v_m(m)}(\sigma^{(N-m)}_m) \ket{\Omega^{(N-m-1)}}.
\end{equation}
It is easy to verify that
\[
\begin{aligned}
&\bra{\omega_{m+1} w_{m+1}^{-1}} B^{(N-m-1)}_{v_m(1)}(\sigma^{(N-m)}_1) \cdots B^{(N-m-1)}_{v_m(m)}(\sigma^{(N-m)}_m) \ket{\Omega^{(N-m-1)}}\\
=& \bar{\partial}^{\beta}_{v^{-1}_m \omega_m} \bra{\omega_{m+1} w_{m+1}^{-1}}B^{(N-m-1)}_{m}(\sigma^{(N-m)}_1) \cdots B^{(N-m-1)}_{1}(\sigma^{(N-m)}_m) \ket{\Omega^{(N-m-1)}}\\
=& \bar{\partial}^{\beta}_{v^{-1}_m \omega_m} \mathcal{G}^{(\beta)}_{w_{m+1}}(\sigma^{(N-m)}_1,\cdots,\sigma^{(N-m)}_m;\ominus \sigma_1^{(N-m-1)},\cdots,\ominus \sigma_m^{(N-m-1)}),
\end{aligned}
\]
where the first equality is due to Lemma \ref{l5} and the second equality is due to Theorem \ref{thm1}.
As $\braket{v_m | \psi^{(N-m)}} = \mathcal{G}^{(\beta)}_{v_m^{-1} \omega_m}(x_1,\cdots,x_{m-1};\ominus \sigma^{(N-m)}_1,\cdots,\ominus \sigma^{(N-m)}_{m-1})$ from the inductive hypothesis, Lemma \ref{l7} and Eq. \eqref{psiplus} then yield
\[
\begin{aligned}
&\braket{\omega_{m+1} w_{m+1}^{-1}| \psi^{(N-m-1)}}\\
=&\sum_{v_m \in S_m} \braket{v_m | \psi^{(N-m)}} \bra{\omega_{m+1} w_{m+1}^{-1}} B^{(N-m-1)}_{v_m(1)}(\sigma^{(N-m)}_1) \cdots B^{(N-m-1)}_{v_m(m)}(\sigma^{(N-m)}_m) \ket{\Omega^{(N-m-1)}}\\
=&\sum_{v_m\in S_{m}} \mathcal{G}^{(\beta)}_{v^{-1}_m \omega_m}(x_1,\cdots,x_{m-1};\ominus\sigma^{(N-m)}_1,\cdots,\ominus \sigma^{(N-m)}_{m-1}) \\
&\quad \quad \quad \cdot \bar{\partial}^\beta_{v_m^{-1} \omega_m} \mathcal{G}^{(\beta)}_{w_{m+1}}(\sigma^{(N-m)}_1,\cdots,\sigma^{(N-m)}_{m};\ominus \sigma^{(N-m-1)}_1,\cdots,\ominus \sigma^{(N-m-1)}_{m})\\
=&\mathcal{G}^{(\beta)}_{w_{m+1}} (x_1,\cdots,x_{m};\ominus \sigma^{(N-m-1)}_1,\cdots,\ominus \sigma^{(N-m-1)}_{m}).
\end{aligned}
\]
Notice that $\sigma^{(0)}_i = t_i$, the proof is thus complete by induction.
\end{proof}

\begin{remark}
As discussed in Section \ref{sec:QK}, when the Bethe Ansatz equations Eq. \eqref{BAE1}-\eqref{BAE3} are taken into account, we should interpret $\sigma^{(N-i)}_a$ ($\beta = 0$) or $1-\sigma^{(N-i)}_a$ ($\beta = -1$), $a=1,\cdots, i$, as the Chern roots of $\mathcal{S}_i$. Therefore, $x_i$ should be interpreted as the first Chern class of $\mathcal{S}_{i}/\mathcal{S}_{i-1}$ when $\beta = 0$ in the cohomology ring, and $1-x_i$ should be interpreted as the K-theory class of $\mathcal{S}_{i}/\mathcal{S}_{i-1}$ when $\beta = -1$ in the K-theory ring\footnote{The identifications imposed in Theorem \ref{thm2} is equivalent to the ring relations with $q_i=0$. For arbitrary $q_i$, we conjecture that the expansion coefficients in Eq. \eqref{thm2exp} should be generalized to the quantum double $\beta$-Grothendieck polynomials. We leave the study of this issue to future work.}. 
\end{remark}
\begin{remark}
More generally, for $\beta \neq 0$, we expect $1+\beta x_i$ to represent $\mathcal{S}_{i}/\mathcal{S}_{i-1}$ in the connective K-theory ring of the flag variety \cite{H13}.
\end{remark}
\begin{remark}
In the case of complete flag variety, $i_a \neq i_b$ for $a \neq b$ in Eq. \eqref{psis}. In the case of general partial flag variety $\mathrm{Fl}(k_{n-1},k_{n-2},\cdots,k_1;N)$, it is possible that $i_a = i_b$ for some $a \neq b$. Because of $[B_i(x),B_i(y)]=0$ (Eq. \eqref{comme}), we expect the expansion coefficients of the Bethe Ansatz state in the natural basis to be given by the double $\beta$-Grothendieck polynomials 
symmetric in $(x_{k_s+1},\cdots,x_{k_{s-1}})$, which can be identified with the Chern roots of $\mathcal{S}_{n-s+1}/\mathcal{S}_{n-s}$, for all $s=1,\cdots,n$.
\end{remark}

\section*{Acknowledgement}
The author would like to thank Xiang-Mao Ding, Wei Gu, Irit Huq-Kuruvilla, Leonardo Mihalcea, Eric Sharpe and Hao Zou for useful comments. This work is partially supported by the National Natural Science Foundation of China (Grant No. 12475005), the Natural Science Foundation of Shanghai Municipality (Grant No. 24ZR1468600), and the Fundamental Research Funds for the Central Universities.

\appendix
\section{Quantum cohomology/K-theory of flag varieties and the double $\beta$-Grothendieck polynomials}\label{app}

In this appendix, we give a brief review of the quantum cohomology/K-theory ring of flag varieties and the double $\beta$-Grothendieck polynomials. The reader may refer to e.g.  \cite{AS95, K95, KM96, FK93, LT06, Buch, Lee:2001mb, MNS1, MNS2} and references therein for more information.

Given $E=\mathbb{C}^N$ and $\bm{a}=(a_1,a_2,\cdots,a_m) \in \mathbb{Z}^m$ such that $0 < a_1 < a_2 < \cdots < a_m < N$, the flag variety $\mathrm{Fl}(\bm{a};N)=\mathrm{Fl}(a_1,a_2,\cdots,a_m;N)$ is defined to be the set of the flags
\begin{equation}\label{flagdef}
V_1 \subset V_2 \subset \cdots \subset V_m \subset E
\end{equation}
such that $\dim(V_i) = a_i$ for all $i$. 
There is a sequence of vector bundles, called the tautological bundles, on $\mathrm{Fl}(a_1,a_2,\cdots,a_m;N)$:
\begin{equation}
\mathcal{S}_1 \subset \mathcal{S}_2 \subset \cdots \subset \mathcal{S}_m \subset \mathrm{Fl}(\bm{a};N) \times E
\end{equation}
with $\mathrm{rank}(\mathcal{S}_i)=a_i$. The fibre of $\mathcal{S}_i$ at the flag \eqref{flagdef} is $V_i$.

Let $S_N$ be the group of permutations of $N$ elements, and let $l(w)$ be the length\footnote{The length of a permutation $w \in S_N$ is the number of pairs $(i,j)$ such that $i<j$ and $w(i) > w(j)$.} of a permutation $w \in \mathrm{S}_N$. The Schubert varieties are subvarieties of $\mathrm{Fl}(\bm{a};N)$ indexed by the set $S_N(\bm{a})=S_N/W_{\bm{a}}$, where $W_{\bm{a}} \subset S_N$ is the subgroup generated by the simple transpositions $s_i = (i,i+1)$ for $ i \not\in \{ a_1,\cdots, a_m \}$. Given a fixed complete flag $F_1 \subset F_2 \subset \cdots \subset F_{N-1} \subset E$, i.e. $\mathrm{dim}(F_i) = i$, and $w \in S_N$ with descents in $\{ a_1, a_2, \cdots, a_n \}$, the Schubert variety $\Omega_w(F_{\bullet})$ is defined to be
\[
\Omega_w(F_{\bullet}) = 
\{ V_{\bullet} \in \mathrm{Fl}(\bm{a};E)~|~ \dim(V_i \cap F_p) \geq \#
\{ t \leq a_i : w(t) > N-p \}, \forall i,p \}.
\]
The codimension of the Schubert variety $\Omega_w(F_{\bullet})$
is $l(w)$.
Let $\Omega^{(\bm{a})}_w$ denote the cohomology class dual to the Schubert variety indexed by $w$, then the Schubert classes $\Omega^{(\bm{a})}_w$ form a basis for the cohomology ring $H^*(\mathrm{Fl}(\bm{a};N),\mathbb{Z})$.

Let $\mathrm{Fl}(N) := \mathrm{Fl}(1,2,\cdots,N-1;N)$ denote the complete flag variety. The cohomology of $\mathrm{Fl}(N)$ has the following presentation
\[
H^*(\mathrm{Fl}(N)) = \mathbb{Z}[x_1,\cdots,x_N]/(e^N_1,\cdots,e^N_N),
\]
where $e^N_i = e_i(x_1,\cdots,x_N)$ is the $i$th elementary symmetric polynomial in $N$ variables. In this presentation, $x_i$ is identified with the Chern class $c_1(\mathcal{S}_i/\mathcal{S}_{i-1})$ and the Schubert class $\Omega_w$ is represented by the Schubert polynomial $X_w (x_1,\cdots,x_N)$, which means that the multiplication rules of the Schubert classes are the same as those of the Schubert polynomials modulo $e^N_i, i=1,\cdots,N$. For general $\bm{a}$, $\Omega^{(\bm{a})}_w$ is represented by the Schubert polynomial $X_w (x_1,\cdots,x_{N-1})$ symmetric in $\{ x_{a_{p-1}+1},\cdots,x_{a_p} \}$ for all $1 \leq p \leq l+1$.

The quantum elementary polynomials $E^N_i$ are defined via the Givental-Kim determinant. Let
\[
\Gamma_N = \left( \begin{array}{ccccc}
x_1 & q_1 & 0 & \cdots & 0 \\
-1 & x_2 & q_2 & \cdots & 0 \\
0 & -1 & x_3 & \cdots & 0 \\
\vdots & \vdots & \vdots & \ddots & \vdots \\
0 & 0 & 0 & \cdots & x_N
\end{array} \right),
\]
then $E^N_i$ is defined as the coefficient of $\lambda^i$ in the characteristic polynomial $\det(1+\lambda \Gamma_N)$.
The quantum cohomology ring of the complete flag variety can be presented as
\[
QH^*(\mathrm{Fl}(N)) = \mathbb{Z}[x_1,\cdots,x_N]/(E^N_1,\cdots,E^N_N).
\]
Ring relations of the quantum cohomology of partial flag varieties $QH^*(\mathrm{Fl}(\bm{a};N))$ can be presented similarly, see \cite{AS95, K95, Buch}. 

The structure constants of $QH^*(\mathrm{Fl}(\bm{a};N))$ encode the Gromov-Witten invariants of the flag variety. The quantum product $*$ of two Schubert classes can be written as
\[
\Omega^{(\bm{a})}_u * \Omega^{(\bm{a})}_v = \sum_{w \in S_N(\bm a)} \sum_{d \in H_2(\mathrm{Fl}(\bm{a};N),\mathbb{Z})} q^d \langle \Omega^{(\bm{a})}_u, \Omega^{(\bm{a})}_v, \Omega^{(\bm{a})}_{\omega_N w \omega_{\bm a}} \rangle_{d} \Omega^{(\bm{a})}_{w},
\]
where $\omega_N$ is the longest permutation in $S_N$, $w_{\bm a}(j) = a_i+a_{i+1}+1-j$ for $a_i < j \leq a_{i+1}$, and $\langle \Omega^{(\bm{a})}_u, \Omega^{(\bm{a})}_v, \Omega^{(\bm{a})}_{w} \rangle_{d}$ is the 3-point, genus-zero Gromov-Witten invariant of degree $d$.

The K-theory ring of the flag variety $K(\mathrm{Fl}(\bm{a};N))$ is the Grothendieck group of the coherent sheaves of $\mathrm{Fl}(\bm{a};N)$. A basis of $K(\mathrm{Fl}(\bm{a};N))$ can be chosen as the structure sheaves of the Schubert varieties $\{\mathcal{O}_w | w \in S_n(\bm{a}) \}$. These Schubert classes are represented by the Grothendieck polynomials $G_w(x_1,\cdots,x_n)$ symmetric in $\{ x_{a_{p-1}+1},\cdots,x_{a_p} \}$ for all $1 \leq p \leq l+1$. Ring relations can be found in e.g. \cite{MNS1}. One way to present the ring relations of the (quantum) K-theory ring of $\mathrm{Fl}(\bm{a};N)$ is to encapsulate them in the (quantum) Whitney relations \cite{Huq-Kuruvilla:2024tsg, GMSXZZ, GMSZ, Gu:2023fpw}.

The structure constants of the quantum K-theory ring $QK(\mathrm{Fl}(\bm{a};N))$ encode the K-theoretic Gromov-Witten invariants \cite{Lee:2001mb} of the flag variety. Let $\langle \mathcal{O}^{(\bm{a})}_{w_1}, \cdots, \mathcal{O}^{(\bm{a})}_{w_m} \rangle_{d}$ denote the $m$-point, genus zero K-theoretic Gromov-Witten invariant of degree $d$. The quantum K-theory ring is a deformation of the K-theory ring whose multiplication $\star$ is defined by
\[
(\mathcal{O}^{(\bm{a})}_u \star \mathcal{O}^{(\bm{a})}_v, \mathcal{O}^{(\bm{a})}_w) = \sum_{d \in H_2(\mathrm{Fl}(\bm{a};N),\mathbb{Z})} q^d \langle \mathcal{O}^{(\bm{a})}_u, \mathcal{O}^{(\bm{a})}_v, \mathcal{O}^{(\bm{a})}_{w} \rangle_{d},
\]
where the pairing is defined by $(\mathcal{O}^{(\bm{a})}_u, \mathcal{O}^{(\bm{a})}_v) = \sum_d q^d \langle \mathcal{O}^{(\bm{a})}_u, \mathcal{O}^{(\bm{a})}_v \rangle_d$.

For the $T=(\mathbb{C}^*)^N$ action on $E=\mathbb{C}^N$ with characters (equivariant parameters) $(t_1,\cdots,t_N)$, one can define the equivariant quantum cohomology ring $QH_T^*(\mathrm{Fl}(\bm{a};N))$ and equivariant quantum K-theory ring $QK_T(\mathrm{Fl}(\bm{a};N))$.
The representatives of the Schubert classes in the equivariant cohomology and K-theory rings are given by the double Schubert polynomials $X_w(x_1,\cdots,x_N;t_1,\cdots,t_N)$ and the double Grothendieck polynomials $G_w(x_1,\cdots,x_N;t_1,\cdots,t_N)$ respectively.

The double Schubert and double Grothendieck polynomials are special cases of a one-parameter family of polynomials called the double $\beta$-Grothedieck polynomials $\mathcal{G}^{(\beta)}_w, w\in S_N$ \cite{FK93}. Let $\mathbb{K}$ be a field of zero characteristic, and let $\beta$ be a formal variable, $\bm{x}=(x_1,\cdots,x_N), \bm{y}=(y_1,\cdots,y_N)$. Define the $\beta$-divided-difference operator $\partial^\beta_i$ acting on $\mathbb{K}(\beta)[x_1,\cdots,x_N,y_1,\cdots,y_N]$ by
\begin{equation}\label{didi}
\partial^{\beta}_i f(x_1,\cdots,x_N;\bm{y}) := \frac{(1+\beta x_{i+1})f(x_1,\cdots,x_N;\bm{y})-(1+\beta x_i)f(x_1,\cdots,x_{i+1},x_i,\cdots,x_N;\bm{y})}{x_i-x_{i+1}}.
\end{equation}
If $w=\omega_N$ is the permutation in $S_N$ with maximal length, then we define
\begin{equation}\label{highestGroth}
\mathcal{G}^{(\beta)}_{\omega_N} (\bm{x};\bm{y})= \prod_{i+j \leq N} (x_i + y_j + \beta x_i y_j),
\end{equation}
and for other $w \in \mathrm{S}_N$, $\mathcal{G}^{(\beta)}_w$ can be defined recursively by
\begin{equation}\label{otherGroth}
\mathcal{G}^{(\beta)}_{w s_i}(\bm{x};\bm{y}) = \partial^{\beta}_i \mathcal{G}^{(\beta)}_{w}(\bm{x};\bm{y})
\end{equation}
if $l(w s_i) = l(w) - 1$. 

The double $\beta$-Grothendieck polynomial $\mathcal{G}^{(\beta)}_{w}(\bm{x};\bm{y})$ reduces to the double Schubert polynomial $X_{w}(\bm{x};\bm{y})$ when $\beta=0$, and reduces to the double Grothendieck polynomial $G_{w}(\bm{x};\bm{y})$ when $\beta = -1$. 

The double $\beta$-Grothendieck polynomials satisfy (see Lemma 5.8 of \cite{H13})
\begin{equation}\label{Gswitch}
\mathcal{G}^{(\beta)}_w(\bm{x};\bm{y}) = \mathcal{G}^{(\beta)}_{w^{-1}}(\bm{y};\bm{x}).
\end{equation}

The dual double $\beta$-Grothendieck polynomials $\mathcal{H}^{(\beta)}_w(\bm{x};\bm{y})$ can be defined recursively by setting $\mathcal{H}^{(\beta)}_{w_N}(\bm{x};\bm{y}) = \mathcal{G}^{(\beta)}_{w_N}(\bm{x};\bm{y})$ and 
\[
\mathcal{H}^{(\beta)}_{w s_i}(\bm{x};\bm{y}) = \bar{\partial}^{\beta}_i \mathcal{H}^{(\beta)}_{w}(\bm{x};\bm{y})
\]
if $l(w s_i) = l(w) - 1$ \cite{LS83}, where $\bar{\partial}^{\beta}_i = \partial^{\beta}_i + \beta$.

More generally, for $v,w \in S_{N}$, one can define the biaxial double $\beta$-Grothendieck polynomials $\mathcal{G}^{(\beta)}_{v,w}(\bm{x};\bm{y})$ as follows (see Definition 2.6 of \cite{BFHTW}):
\begin{equation}\label{dG1}
\mathcal{G}^{(\beta)}_{1,w}(\bm{x};\bm{y}) = \mathcal{G}^{(\beta)}_{w}(\bm{x};\bm{y}).
\end{equation}
If $l(v s_i) = l(v) +1$, set 
\begin{equation}\label{dG2}
\mathcal{G}^{(\beta)}_{v s_i,w}(\bm{x};\bm{y}) = \bar{\partial}^\beta_{i, \bm{y}} \mathcal{G}^{(\beta)}_{v,w}(\bm{x};\bm{y}), 
\end{equation}
where $\bar{\partial}^\beta_{i, \bm{y}}$ is the $\bar{\partial}^\beta_{i}$ operator acting on the set of variables $\bm{y} = (y_1,y_2,\cdots,y_N)$.
We have the following generalized Cauchy identity:
\begin{thm}[Theorem 7.1 of \cite{BFHTW}] For any $w \in S_N$,
\begin{equation}\label{Gcauchy}
\mathcal{G}^{(\beta)}_{w}(\bm{x};\bm{y}) = \sum_{v \in S_N} \mathcal{G}^{(\beta)}_{v}(\bm{y};\bm{z}) \mathcal{G}^{(\beta)}_{v^{-1},w}(\bm{x};\ominus\bm{z}),
\end{equation}
where $\ominus z_i = -z_i/(1+\beta z_i)$.
\end{thm}

\end{document}